\documentclass[american,aps,pra,reprint, superscriptaddress]{revtex4-1}
\usepackage{amsmath,amscd,amsthm,amssymb}
\usepackage{mathrsfs}
\usepackage{graphicx}
\usepackage{epsf,epstopdf}
\usepackage[center]{caption2}
\usepackage[all]{xy}
\usepackage[unicode=true,pdfusetitle, bookmarks=true,bookmarksnumbered=false,bookmarksopen=false, breaklinks=false,pdfborder={0 0 0},backref=false,colorlinks=false] {hyperref}
\hypersetup{colorlinks=true,linkcolor=myurlcolor,citecolor=myurlcolor,urlcolor=myurlcolor}
\usepackage{graphics,epstopdf,graphicx, amsthm, amsmath, amssymb, times, braket, color, bm}
\usepackage[up]{subfigure}
\usepackage{cleveref}

\usepackage{makecell}

\definecolor{myurlcolor}{rgb}{0,0,0.7}

\theoremstyle{plain}
\newtheorem{thm}{\protect\theoremname}
\newtheorem{prop}[thm]{Proposition}
\newtheorem{lem}[thm]{Lemma}
\providecommand{\theoremname}{Theorem}
\newcommand*{\myproofname}{Proof}

\makeatother

\newtheorem*{cor}{Corollary}

\theoremstyle{definition}

\theoremstyle{remark}
\newtheorem{rem}{Remark}

\begin{document}

 \author{Chunhe Xiong}
 \email{xiongchunhe@zju.edu.cn}
 \affiliation{School of Mathematical Sciences, Zhejiang University, Hangzhou 310027, PR~China}
 \affiliation{Department of Computer Science, Sun Yat-sen University, Guangzhou 510006, PR~China}

 \author{Asutosh Kumar}
 \email{Corresponding author: asutoshk.phys@gmail.com}
 \affiliation{P.G. Department of Physics, Gaya College, Magadh University, Rampur, Gaya 823001, India}

 \author{Minyi Huang}
 \email{Corresponding author: hmyzd2011@126.com}
 \affiliation{School of Mathematical Sciences, Zhejiang University, Hangzhou 310027, PR~China}
 \affiliation{Department of Mathematics, Zhejiang Sci-Tech University, Hangzhou 310018, PR~China}

 \author{Sreetama Das}
 \email{sreetamadas@hri.res.in}
 \affiliation{Harish-Chandra Research Institute, HBNI, Chhatnag Road, Jhunsi, Allahabad 211019, India}

  \author{Ujjwal Sen}
 \email{Corresponding author: ujjwal@hri.res.in}
 \affiliation{Harish-Chandra Research Institute, HBNI, Chhatnag Road, Jhunsi, Allahabad 211019, India}

 \author{Junde Wu}
 \email{Corresponding author: wjd@zju.edu.cn}
 \affiliation{School of Mathematical Sciences, Zhejiang University, Hangzhou 310027, PR~China}

\title{Partial coherence and quantum correlation with fidelity and affinity distances}
\begin{abstract}
A fundamental task in any physical theory is to quantify certain physical quantity in a meaningful way.
In this paper we show that both fidelity distance and affinity distance satisfy the strong contractibility, and the corresponding resource quantifiers can be used to characterize a large class of resource theories.
Under two assumptions, namely, convexity of ``free states'' and closure of free states under ``selective free operations'', our general framework of resource theory includes quantum resource theories of entanglement, coherence, partial coherence and superposition.
In partial coherence theory, we show that fidelity partial coherence of a bipartite state is equal to the minimal error probability of a mixed quantum state discrimination (QSD) task and vice versa, which complements the main result in [\href{http://iopscience.iop.org/10.1088/1751-8121/aac979}{Xiong and Wu, J. Phys. A: Math. Theor. {\bf 51}, 414005 (2018)}]. 
We also compute the analytic expression of fidelity partial coherence for $(2,n)$ bipartite X-states. 
At last, we study the correlated coherence in the framework of partial coherence theory. We show that partial coherence of a bipartite state, with respect to the eigenbasis of a subsystem, is actually a measure of quantum correlation.
\end{abstract}

\maketitle

\section{introduction}

Quantum correlations \cite{Horodecki2009R, Modi2012}
have been shown to be an important resource in quantum information theory (QIT) \cite{nielsen10} since they often offer a remarkable advantage in information processing tasks over classical theory. It has been argued that there may be several independent resources that provide a quantum advantage, including quantum entanglement \cite{Horodecki2009R}, quantum coherence \cite{Aberg2006,Baumgratz2014,Winter2016,streltsov2017A}, asymmetry \cite{Marvian2013,Marvian2014}, athermality \cite{Brandao2015,Horodecki2013} and quantum superposition \cite{Theurer2017}, etc.
Every resource theory puts certain restrictions on quantum states and quantum operations in the sense that what states and operations are accessible ``free of cost'' and what are ``assets'' or resource. This however does not mean that free states and free operations of a resource theory do not incur costs of preparation and implementation.

Thus, any resource theory begins with the identification of ``free states'' that do not possess resource, and ``free operations'' that do not generate resource.  We will denote the set of free states by $\mathcal{FS}$ and the set of free operations by $\mathcal{FO}$. By definition, a state outside $\mathcal{FS}$ and an operation outside $\mathcal{FO}$ are regarded as resources. Quantifying the resource for a given state is a fundamental problem in the resource theory.  Particularly, 
geometric entanglement and geometric coherence which characterize the minimal distance of the state to the free states, have been proposed as measures of entanglement and coherence respectively \cite{Wei2003,Aberg2006,Baumgratz2014, Winter2016,streltsov2017A}.
Recently, a common framework for characterization and quantification of the convex quantum resources was proposed in \cite{Regula2018A}. It was shown that, for generic resource theories, one can establish the strong monotonicity of a family of resource quantifiers such as robustness monotones and norm-based quantifiers, etc.
%


 In this paper, we first prove that both fidelity distance and affinity distance also satisfy the strong monotonicity condition and can be used to characterize generic resource theories, including entanglement, coherence, partial coherence and superposition.

Next, we focus on partial coherence theory which is an extension of coherence theory \cite{Luo2017A,Kim2018B}. From the viewpoint of partial coherence, quantum discord can be regarded as the minimal partial coherence over all local orthogonal basis \cite{Luo2017B}. In Ref. \cite{spehner2013A}, the authors linked quantum discord with QSD for the first time. An equivalence between pure state QSD task and coherence theory was established in Ref. \cite{Xiong2018A}. It is a difficult problem to obtain this equivalence for general mixed states. However, by constructing a QSD-state for each QSD task, we will show that the QSD task for mixed states just corresponds to partial coherence. In this way, we offer the operational meaning for both fidelity- and affinity-based partial coherence measures. Our results thus establish a useful connection between partial coherence theory and QSD.

Finally, we compute the analytic expression of fidelity partial coherence for bipartite X-states.



In Sec. \ref{sec:setting the stage}, we provide a few definitions and notations.  Fidelity and affinity and the corresponding distance measures are defined in Sec. \ref{sec:quantifying resource} A and \ref{sec:quantifying resource} B, and prove the strong monotonicity of the distance measures in the latter. In Sec. \ref{sec:quantifying resource} C, we quantify resource based on fidelity and affinity, and prove some of their properties. The formalism developed in these sections is employed in Sec. \ref{fidelity-coherence} to study the quantification of partial coherence with fidelity distance and its connection with QSD. The connection between affinity partial coherence and QSD is discussed in Sec. \ref{affinity-coherence}. In Sec. \ref{sec:correlated coherence}, we study the correlated coherence in the framework of partial coherence theory. We present a conclusion in Sec. \ref{sec:conclusion}.

\section{setting the stage}
\label{sec:setting the stage}

In this paper we consider a general framework of resource theory for finite dimensional quantum systems which is built on two postulates: convexity of free states and closure of free states under ``selective'' free operations.

Let $\rho,\sigma\in\mathcal{FS}$ and $p\in(0,1)$. Then it is natural to ask whether $p\rho+(1-p)\sigma\in\mathcal{FS}$.
We adopt the following postulate for free states.
%

{\bf Postulate 1.} $\mathcal{FS}$ is convex, meaning that linear convex combinations of free states are also free states.\\

Free operations will transform free states to free states. In practice, one may also demand that selective measurements (measurements whose outcomes are separately accessible) map free states to free ones. In such a case, there exists some Kraus decomposition $\{K_n\}$ of a free operation $\Phi\in\mathcal{FO}$ such that 
$K_n\sigma K^{\dagger}_n\in\mathcal{FS}$ for each $n$ and $\sigma\in\mathcal{FS}$. This means that we have adopted the following postulate for free operations.

{\bf Postulate 2.} Each free operation $\Phi\in\mathcal{FO}$ admits a Kraus representation $\Phi(\cdot)=\sum_n\Phi_n(\cdot)=\sum_nK_n(\cdot) K^{\dagger}_n$, such that $K_n\mathcal{FS}K^{\dagger}_n\subseteq\mathcal{FS}$ for each $n$.

\begin{rem}
These two requirements are natural and practical, and resource theories satisfying these two postulates include entanglement, quantum coherence, and superposition. In case of Postulate 1, one may additionally demand that $\mathcal{FS}$ is closed, i.e., it contains all its limit points.
\end{rem}

A good resource measure must vanish for free states since these states are regarded as no-resource states. Moreover, as free operations can be executed generously, the resource of a state should not increase under free operations. As a result, these two conditions are fundamental for a resource quantifier $R$ and can be expressed mathematically as follows:

(R1') {\it Non-negativity}. $R(\rho)\ge0$ and $R(\rho)=0$ for every $\rho\in\mathcal{FS}$.

(R1) {\it Faithfulness}. $R(\rho)\ge0$ and equality holds if and only if $\rho\in\mathcal{FS}$.

\begin{rem}
Note that (R1') is a weaker condition than (R1). According to (R1'), resource measures must vanish, by definition, for each state in $\mathcal{FS}$. In additon, these resource measures might also vanish for certain states that do not belong to $\mathcal{FS}$. This leaves  room for operational measures like distillable entanglement \cite{Bennett1996,Horodecki1998} which vanishes for certain entangled states \cite{Horodecki2009R}. On the other hand, (R1) says that resource measures will vanish only for states in $\mathcal{FS}$.
\end{rem}

(R2) {\it Monotonicity}. $R(\rho)\ge R(\Phi(\rho))$ for any $\rho$ and $\Phi\in\mathcal{FO}$.

Parallel to Postulate 2, we may also require that the resource of a state does not increase under selective measurements. This translates into the strong monotonicity condition below.\\
(R3) {\it Strong monotonicity}. $R(\rho) \ge \sum_n p_n R(\rho_n)$, where $p_n={\mathrm{tr}}(K_n\rho K^{\dagger}_n)$ and $\rho_n=\frac{K_n\rho K^{\dagger}_n}{p_n}$, for a free operation $\Phi \equiv \{K_n\}\in\mathcal{FO}$ such that $\sum_nK^{\dagger}_nK_n=I$ and $K_n\mathcal{FS}K^{\dagger}_n\subseteq\mathcal{FS}$.\\

Implicit to the structure of a resource theory, there are two natural ways to quantify resource. If the resource theory has a ``unit resource'', such as an ebit in entanglement theory \cite{Bennett1996} and the maximally coherent pure state in coherence theory \cite{Winter2016}, distillable resource and cost of resource can be considered \cite{Rains1999A,Rains1999B}. Irrespective of such possibilities of defining resources via rates, distances in state space also provide excellent avenues to quantify resource. Supposing $d$ is some distance in the state space, the resource measure can be defined as
\begin{align}\label{eq:resource-measure}
R_d(\rho):=\min_{\sigma\in\mathcal{FS}}d(\rho,\sigma).
\end{align}

It is clear that $R_d$ satisfies (R1), provided $\mathcal{FS}$ is closed. Moreover, if the distance is contractive, i.e., $d(\rho,\sigma)\ge d\left(\Psi(\rho),\Psi(\sigma)\right)$ for any completely-positive and trace-preserving (CPTP) map $\Psi$, then $R_d$ satisfies (R2) also.
%
We say a distance $d$ satisfies {\it strong contractibility} if
\begin{align}\label{eq:strong-monotonicity-distance}
\sum_ip_id(\rho_i,\sigma_i)\le d(\rho,\sigma),
\end{align}
for any pair of density matrices $\rho$ and $\sigma$, a set of Kraus operators $\{K_i\}$ and  $p_i=\mathrm{Tr}(K_i\rho K^{\dagger}_i)$, $q_i=\mathrm{Tr}(K_i\sigma K^{\dagger}_i)$, $\rho_i=\frac{K_i\rho K^{\dagger}_i}{p_i}$,  $\sigma_i=\frac{K_i\sigma K^{\dagger}_i}{q_i}$.
%
Supposing distance $d$ satisfies strong contractibility, we have
\begin{align*}
R_d(\rho)=d(\rho,\sigma^{\star})\overset{Eq. (\ref{eq:strong-monotonicity-distance})} \ge \sum_np_nd(\rho_n,\sigma^{\star}_n) \overset{Eq. (\ref{eq:resource-measure})} \ge\sum_np_nR_d(\rho_n),
\end{align*}
where $\sigma^{\star}$ is the free state for which minimum is achieved in Eq. (\ref{eq:resource-measure}), $\sigma^{\star}_n=K_n\sigma^{\star}K^{\dagger}_n/\mathrm{tr}(K_n\sigma^{\star}K^{\dagger}_n)\in\mathcal{FS}$ and other notations are as defined above.
While the first inequality is due to the strong contractibility of $d$, the second inequality follows from the definition of $R_d$.
We may also demand the resource quantifier to obey convexity. That is, resource should not increase under convex combination of quantum states.

(R4) {\it Convexity}. If the system is in state $\rho_i$ with probability $p_i$, then $R\left(\sum_ip_i\rho_i\right)\le\sum_ip_iR(\rho_i)$.\\

Similar to entanglement, a quantifier $R$ in this general resource theory is called (strong) resource monotone if it satisfies (strong) monotonicity and faithfulness. In addition, if $R$ satisfies convexity also, we call it convex (strong) resource monotone. A convex strong resource monotone is called a resource measure.

\section{quantifying resource}
\label{sec:quantifying resource}
In this section, we review two distance measures using which we can establish bona fide measures to quantify a general resource.

\subsection{fidelity and affinity}

Let $\mathcal{H}$ be an $n$-dimensional Hilbert space and $\mathcal{E}(\mathcal{H})$ be the set of density matrices on $\mathcal{H}$. For any $\rho,\sigma\in\mathcal{E}(\mathcal{H})$, the fidelity between $\rho$ and $\sigma$ is defined as \cite{nielsen10}:
\begin{align}\label{eq6}
F(\rho,\sigma):=||\sqrt{\rho}\sqrt{\sigma}||_1=\mathrm{Tr}\sqrt{\sqrt{\sigma}\rho\sqrt{\sigma}}.
\end{align}
Similarly, quantum affinity is defined as follows \cite{luo2004}:
 \begin{align}\label{eq5}
 A(\rho,\sigma):=\mathrm{Tr}(\sqrt{\rho}\sqrt{\sigma}).
 \end{align}
This definition is similar to the Bhattacharyya coefficient \cite{Bhattacharyya1943} between two probability distributions (discrete or continuous) in classical probability theory. Both fidelity and affinity characterize the closeness of two quantum states in the state space.

As $F(\rho,\sigma)=\max_U|\mathrm{Tr}U\sqrt{\rho}\sqrt{\sigma}|$, with the maximization being over all unitary operators on $\mathcal{H}$, we have $A(\rho,\sigma)\le F(\rho,\sigma)$.
Moreover, since $\mathrm{Tr}(\sqrt{\rho}\sqrt{\sigma})=\mathrm{Tr}(\rho^{1/4}\sigma^{1/2}\rho^{1/4})$, we have $0\le A(\rho,\sigma)\le 1$,
and $A(\rho,\sigma)=1$ if and only if $\rho=\sigma$.

Let $X$ be fidelity or affinity. Then $X$ satisfies the following properties:

 (P1) $X(\rho,\sigma)\in[0,1]$ with $X(\rho,\sigma)=1$ iff $\rho=\sigma$.

 (P2) $X(\frac{\rho}{p},\frac{\sigma}{q})=\frac{X(\rho,\sigma)}{\sqrt{pq}}$ with $p,q\in(0,1)$.

 (P3) $X(\rho,\sigma)$ obeys monotonicity under CPTP maps.

 (P4) $X(\sum_i P_i\rho P_i,\sum_i P_i\sigma P_i)=\sum_i X(P_i\rho P_i,P_i\sigma P_i)$ for mutually orthogonal projectors $\{P_i\}$.
\begin{proof}
Since $\{P_i\}$ are mutually orthogonal projectors, we have $\sqrt{\sum_iP_i\rho P_i}=\sum_i\sqrt{P_i\rho P_i}$ and
 \begin{align}
  A\left(\sum_i P_i\rho P_i,\sum_i P_i\sigma P_i \right)=&\mathrm{Tr}\sqrt{\sum_iP_i\rho P_i}\sqrt{\sum_j P_j\sigma P_j}\nonumber\\
 =&\mathrm{Tr}\sum_i\sqrt{P_i\rho P_i}\sum_j \sqrt{ P_j\sigma P_j}\nonumber\\
 =&\sum_i\mathrm{Tr}\sqrt{P_i\rho P_i}\sqrt{ P_i\sigma P_i}\nonumber\\
 =&\sum_iA(P_i\rho P_i,P_i\sigma P_i).\nonumber
 \end{align}
Similarly, we can show that (P4) holds for fidelity $F$.
\end{proof}

 (P5) For a CPTP map $\Phi \equiv \{K_i\}$,
\begin{align}
X(\rho,\sigma)\le\sum_i X(p_i \rho_i, q_i \sigma_i) = \sum_i X(K_i\rho K_i^{\dagger},K_i\sigma K^{\dagger}_i),
\end{align}
where $\rho_i=\frac{K_i\rho K_i^{\dagger}}{p_i}$ and $\sigma_i=\frac{K_i\sigma K^{\dagger}_i}{q_i}$ are the states after subselection with probabilities $p_i=\mathrm{Tr}(K_i\rho_iK_i^{\dagger})$ and $q_i=\mathrm{Tr}(K_i\sigma_iK_i^{\dagger})$, respectively. Property (P5) can be proven using the method in Ref. \cite{vedral1998} and exploiting properties (P3) and (P4) (see Ref. \cite{note2}).

\subsection{distance measures using fidelity and affinity}

Assuming $X$ to be a function in $\mathcal{E}(\mathcal{H})\otimes\mathcal{E}(\mathcal{H})$ satisfying (P1-P5), the distance given by
 \begin{align}\label{eq7}
 d_X(\rho,\sigma):=1-X^2(\rho,\sigma),
 \end{align}
has the following two properties:

(D1) $d_X(\rho,\sigma)\ge 0$ with equality iff $\rho=\sigma$.

(D2) $d_X(\rho,\sigma)$ is contractive, that is, $d_X(\Phi(\rho),\Phi(\sigma))\le d_X(\rho,\sigma)$ for any CPTP map $\Phi$.

With this, we arrive at our first main result in this paper: $d_X(\rho,\sigma)$ satisfies the strong contractibility. To prove this we need the following lemma.

\begin{lem}\label{lem2}
For a probability vector $\{p_i\}^n_{i=1}$ and a vector of positive real numbers $\{x_i\}^n_{i=1}$,
\begin{align}
\sum_i \frac{x^2_i}{p_i} \ge \left(\sum_i x_i \right)^2.
\label{eq:lem2}
\end{align}
\end{lem}

\begin{proof}
Suppose $\{p_i\}^n_{i=1}$ is a probability vector and $\{x_i\}^n_{i=1}$ are \(n\) positive real numbers. Then
\begin{align}
\sum_i\frac{x^2_i}{p_i}&=\sum_ix^2_i+\sum_i\sum_{j\ne i}\frac{p_j}{p_i}x^2_i \nonumber\\
&=\sum_ix^2_i+\sum_{i<j} \left(\frac{p_j}{p_i}x^2_i+\frac{p_i}{p_j}x^2_j \right) \nonumber\\
&\ge\sum_ix^2_i+2\sum_{i<j}x_ix_j= \left(\sum_ix_i \right)^2. \nonumber
\end{align}
\end{proof}

\begin{thm}\label{thm1}
$d_X(\rho,\sigma)$ satisfies strong contractibility.
\end{thm}

\begin{proof} We have
\begin{align}\label{eq:strong-monotonicity-theorem}
&\sum_ip_id_X(\rho_i,\sigma_i)=1-\sum_ip_i X^2(\rho_i,\sigma_i) \nonumber\\
\overset{(P2)}=&1-\sum_iq^{-1}_i X^2(K_i\rho K_i^{\dagger},K_i\sigma K^{\dagger}_i) \nonumber\\
\overset{Eq. (\ref{eq:lem2})} \le&1- \left(\sum_i X(K_i\rho K_i^{\dagger},K_i\sigma K^{\dagger}_i) \right)^2 \nonumber\\
\overset{(P5)} \le&1-X^2(\rho,\sigma) \nonumber\\
=&d_X(\rho,\sigma),
\end{align}
where  the second equality is due to property (P2), the first inequality follows from Lemma \ref{lem2} and the second inequality is due to property (P5) .
\end{proof}

It is not difficult to show that $d_X$ is also symmetric, i.e., Eq. (\ref{eq:strong-monotonicity-theorem}) holds if we replace $\{p_i\}$ with $\{q_i\}$ (this statement will not hold true for ``asymmetric distances'' like relative entropy \cite{Wehrl1978,Vedral2002}). $d_X$ is also bounded.

We can choose \(X\) to be either fidelity and affinity, and define
fidelity distance
\begin{align*}
d_F(\rho,\sigma):=1-F^2(\rho,\sigma);
\end{align*}
and affinity distance
\begin{align*}
d_A(\rho,\sigma):=1-A^2(\rho,\sigma).
\end{align*}
Both $d_F$ and $d_A$ satisfy the strong contractibility condition.

\begin{rem}
Relative entropy is a useful ``distance'' with several nice properties including strong contractibility \cite{vedral1998} and joint convexity \cite{Lindblad1974}, which means that the corresponding resource quantifier will satisfy properties (R1) through (R4). In fact, as a distance-based quantifier, relative entropy has been proven to be a bona fide resource measure of entanglement, quantum correlations beyond entanglement, quantum coherence, and superposition, and the operational meanings associated with them have been explored \cite{Henderson2000,Winter2016}. These resource theories can be studied from a unified perspective when relative entropy is employed to define a resource measure \cite{Modi2010,Yao2015}.
\end{rem}

\subsection{quantifying resource with fidelity and affinity}

We now employ the fidelity and affinity distances to define the corresponding resource quantifiers as follows:

(i) Fidelity resource, also called geometric resource, based on fidelity distance,
\begin{align}\label{eq2}
R_F(\rho):=\min_{\sigma\in\mathcal{FS}}d_F(\rho,\sigma);
\end{align}

(ii) Affinity resource based on affinity distance,
\begin{align}\label{eq3}
R_A(\rho):=\min_{\sigma\in\mathcal{FS}}d_A(\rho,\sigma).
\end{align}

Since both fidelity and affinity distances observe the strong monotonicity condition, the above resource quantifiers are strong resource monotones.

\begin{thm}\label{thm2}
Fidelity resource $R_F$ satisfies (R1-R4).
\end{thm}
%
%
\begin{proof}
$R_F$ inherits properties (R1-R3) from $d_F$. Hence, we need to prove only convexity (R4) here.
As the maximum of square fidelity between a state $\rho$ and a convex state set is equal to the maximum over mixture of $\rho$ over the state set \cite{streltsov2010}, i.e.,
\begin{align}\label{eq1}
\max_{\sigma\in\mathcal{FS}}F^2(\rho,\sigma)=\max_{\rho=\sum_ip_i\rho_i}\sum_ip_i\max_{\sigma_i\in\mathcal{FS}}F^2(\rho_i,\sigma_i),
\end{align}
where the first maximization on the right-hand side is over all $\rho=\sum_ip_i\rho_i$, fidelity resource can be expressed as follows (see Ref. \cite{note3}):
\begin{align}\label{convex-roof}
R_F(\rho)=\min_{\rho=\sum_ip_i\rho_i}\sum_i p_i R_F(\rho_i).
\end{align}

Suppose $\rho_i=\sum_jq^i_j\rho^i_j$ is the optimal state decomposition for each $\rho_i$ in the sense that Eq. (\ref{convex-roof}) is satisfied. Then, for an arbitrary mixed quantum state $\rho=\sum_i p_i \rho_i \left(= \sum_{i,j}p_iq^i_j\rho^i_j \right)$, we have
\begin{align}
R_F \left(\sum_ip_i\rho_i \right)&=R_F(\rho)=R_F \left(\sum_{i,j}p_iq^i_j\rho^i_j \right) \nonumber \\
&\overset{Eq. (\ref{convex-roof})}\le \sum_{i,j}p_iq^i_j R_F(\rho^i_j) \nonumber \\
&\overset{Eq. (\ref{convex-roof})}=\sum_i p_i R_F(\rho_i). \nonumber
\end{align}

In conclusion, fidelity resource $R_F$ satisfies (R1-R4).
\end{proof}

\begin{rem}\label{rem4}
Affinity resource $R_A$ also satisfies (R1-R3) because of $d_A$.
But, since affinity does not satisfy Eq. (\ref{eq1}) and $d_A$ is not jointly convex, $R_A$ may not fulfil (R4). However, in coherence theory, the measure based on affinity distance satisfies convexity \cite{yucs2017}. On the other hand, the convex roof extension (which extends a resource quantifier for pure states to mixed states) of $R_A$ can be shown to satisfy (R1-R4).

Fidelity and affinity distances based resource measures have been proved to be bona fide resource measures in entanglement \cite{Wei2003} and coherence \cite{Streltsov2015B,yucs2017,Xiong2018B} resource theories.
\end{rem}

\section{fidelity partial coherence and quantum state discrimination} \label{fidelity-coherence}
\subsection{partial coherence theory}

Consider a bipartite quantum system ``ab'' with Hilbert space $H=H_a\otimes H_b$, where $H_a$ and $H_b$ are the Hilbert spaces of the subsystems ``a'' and ``b'' having finite dimensions $n_a$ and $n_b$ respectively. Let $\{\ket{i}\}$ be a fixed orthogonal basis of party ``a'', then $\Pi_L=\{\ket{i}\bra{i}\otimes I^b\}$ is the L{\"u}ders measurement and the notions for partial coherence respect to $\Pi_L$ are as follows:

(1) The set of partial ``incoherent'' or free states are defined by
\begin{align*}
I^a_P=\{\sigma:\Pi_L(\sigma)=\sigma\},
\end{align*}
where $\Pi_L(\sigma)=\sum_i (\Pi^a_i\otimes I^b) \sigma (\Pi^a_i\otimes I^b)$. Denoting $p_i=\mathrm{tr}\bra{i}\sigma\ket{i}$
and $\sigma_i=p^{-1}_i\bra{i}\sigma\ket{i}$ [here $\bra{i}\sigma\ket{i} \equiv(\bra{i}\otimes I_b)\sigma(\ket{i} \otimes I_b)$ for brevity; we observe similar notation elsewhere also], each partial incoherent state can be written as
\begin{align}\label{eq22}
\sigma=\sum_ip_i\ket{i}\bra{i}\otimes\sigma_i.
\end{align}

(2) A CPTP map $\Phi^a$ with Kraus operators $\{K_n\}$ is called partial incoherent if $K_nI^a_PK^{\dagger}_n\in I^a_P$, and the set of partial incoherent operations is denoted as $\mathcal{O}^a_P$.\\

A functional $C^a$ on the bipartite system, satisfying the conditions (C1-C4) below, is called a measure of partial coherence with respect to $\Pi_L$.

(C1) Nonnegativity: $C^a(\rho^{ab})\ge0$, and the equality holds if and only if $\sigma\in I^a_P$.

(C2) Monotonicity under partial incoherent operations: $C^a(\Phi^a(\rho^{ab}))\le C^a(\rho^{ab})$ for all $\Phi^a\in\mathcal{O}^a_P$.

(C3) Monotonicity under selective partial incoherent operations on average: $\sum_ip_iC^a(p^{-1}_iK_i\rho^{ab}K^{\dagger}_i)\le C^a(\rho^{ab})$ with probabilities $p_i=\mathrm{tr}(K_i\rho^{ab}K^{\dagger}_i)$ and partial incoherent Kraus operators $K_i$.

(C4) Convexity:  $C^a(\sum_ip_i\rho^{ab}_i)\le \sum_ip_iC^a(\rho^{ab}_i)$ for any ensemble $\{p_i, \rho^{ab}_i\}$ with $p_i\ge0$ and $\sum_ip_i=1$.

Based on \eqref{eq2} and \eqref{eq3}, we define fidelity (geometric) partial coherence by
\begin{align}\label{eq25}
C^a_F(\rho^{ab}):=\min_{\sigma\in I^a_P}d_F(\rho^{ab},\sigma);
\end{align}
and affinity partial coherence by
\begin{align}\label{eq26}
C^a_A(\rho^{ab}):=\min_{\sigma\in I^a_P}d_A(\rho^{ab},\sigma),
\end{align}
respectively. Theorem \ref{thm2} ensures that $C^a_F$ is a partial coherence measure and $C^a_A$ is a strong partial coherence monotone.

\subsection{quantum state discrimination}


In QSD task, the sender chooses a state randomly from the ensemble $\{\rho_i, \eta_i \}$ and sends it to the receiver, whose job is to determine the received state with maximal probability. To do so, he performs a positive operator valued measurement (POVM) $\{M_i : M_i \ge 0, \sum_i M_i =I\}$ on each $\rho_i$ and declares the state is $\rho_j$ when the measurement reads $j$. As the probability to get the result $j$ is $p_{j|i}=\mathrm{Tr}(M_j\rho_i)$ when the system is in the state $\rho_i$, the maximal success probability to identify $\{\rho_i,\eta_i\}$ is
\begin{align}
P^{opt}_S(\{\rho_i,\eta_i\})=\mathop{\mathrm{max}}_{\{M_i\}}\sum_i\eta_i\mathrm{Tr}(M_i\rho_i),\nonumber
\end{align}
where the maximization is performed over all POVM $\{M_i\}$, and the minimal error probability is
\begin{align}
P^{opt}_E(\{\rho_i,\eta_i\})=1-\mathop{\mathrm{max}}_{\{M_i\}}\sum_i\eta_i\mathrm{Tr}(M_i\rho_i).\nonumber
\end{align}

For an ensemble consisting of two states, the analytic formula of maximal success probability is given by Helstrom formula \cite{Helstrom1976}
\begin{align}\label{eq24}
P^{opt}_S(\{\rho_i,\eta_i\}^2_{i=1})=\frac{1}{2}(1-\mathrm{tr}|\Lambda|),
\end{align}
where $\Lambda=\eta_1\rho_1-\eta_2\rho_2~(=\sum_i \lambda_i \ket{i}\bra{i});~|\Lambda|:=\sum_i |\lambda_i| \ket{i}\bra{i}$, and the corresponding optimal measurement is a von Neumann measurement $\{\Pi^{opt}_1,I-\Pi^{opt}_1\}$ with $\Pi^{opt}_1$ being the projector onto the support of $\Lambda_{+}=(\Lambda+|\Lambda|)/2$. However, no solution is known for general case except some symmetric cases \cite{Bergou2004,Eldar2004,Chou2003}.

As a suboptimal choice, the least square measurement (LSM)  is an alternative to discriminate quantum states \cite{Belavkin1975a,Belavkin1975,Holevo1978,Hausladen1994,Hausladen1996,peres1991,Eldar2001}. In comparison to the optimal measurement, the LSM has several nice properties. First, its construction is relatively simple as it can be determined directly from the given ensemble. Second, it is very close to the optimal measurement when the states to be distinguished are {\it almost orthogonal} \cite{Holevo1978,Spehner2014}. The construction of LSM is as follows.

For the ensemble $\{\rho_i,\eta_i\}^n_{i=1}$, the least square measurements are given by 
\begin{align}\label{eq32}
M^{lsm}_i=\eta_i\rho_{out}^{-1/2}\rho_i\rho_{out}^{-1/2},i=1,2,...,n,
\end{align}
where $\rho_{out}=\sum_i \lambda_i \rho_i$.  As a result, the minimal error probability of this measurement is
\begin{align}\label{eq4}
P^{lsm}_E(\{\rho_i,\eta_i\})=1-\sum_i\eta_i\mathrm{Tr}(M^{lsm}_i\rho_i).
\end{align}

\subsection{fidelity partial coherence}

Assuming $\rho_i=\sum^m_j\lambda_{ij}\ket{\psi_{ij}}\bra{\psi_{ij}}$ is the spectral decomposition of $i$, the density matrices $\{\rho_i, 1\le i\le n\}$ are called {\it linearly independent} if $\{\ket{\psi_{ij}},1\le i\le n, 1\le j\le m \}$ are linearly independent.



It is well known that POVMs can outperform von Neumann measurements in quantum state discrimination task \cite{Peres1990,peres1991}. However, the von Neumann measurement has proved to be the optimal one so far in discriminating a collection of linearly independent states, both pure \cite{kennedy1973} and mixed \cite{Eldar2003}. With this observation, we can establish the relation between fidelity partial coherence and QSD.

\begin{thm}\label{thm4}
For any bipartite state $\rho^{ab}$, the fidelity partial coherence is given by
\begin{align}
C^a_F(\rho^{ab})=P^{opt (vN)}_E(\{\omega_i,\eta_i\}^{n_a}_{i=1}),
\end{align}
where $\omega_i=\eta^{-1}_i\sqrt{\rho^{ab}}\ket{i}\bra{i}\otimes I^b\sqrt{\rho^{ab}}$ with $\eta_i=\mathrm{tr}\bra{i}\rho^{ab}\ket{i}$ and  $P^{opt (vN)}_E(\{\omega_i,\eta_i\}^{n_a}_{i=1})$ is the minimal error probability to discriminate $\{\omega_i,\eta_i\}^{n_a}_{i=1}$ with von Neumann measurement.
%
As a result, fidelity partial coherence provides an upper bound for the minimum error probability to discriminate $\{\omega_i,\eta_i\}^{n_a}_{i=1}$,
\begin{align}\label{eq11}
C^a_F(\rho^{ab})\ge P^{opt}_E(\{\omega_i,\eta_i\}^{n_a}_{i=1}).
\end{align}
In particular, if $\{\omega_i\}^{n_a}_{i=1}$ are linearly independent then $C^a_F(\rho^{ab})$ is exactly the minimum error probability.
\end{thm}

\begin{proof}First, we evaluate fidelity partial coherence of $\rho^{ab}$.

Based on Theorem 3 in Ref. \cite{spehner2013A}, the fidelity between $\rho^{ab}$ and the partial incoherent states is
\begin{align*}
F(\rho^{ab}):&=\max_{\sigma\in I^a_P}F(\rho^{ab},\sigma)\\
&=\max_{\{\pi_i\}}\sqrt{\sum_i\mathrm{tr}[\pi_i\sqrt{\rho^{ab}}\ket{i}\bra{i}\otimes I^b\sqrt{\rho^{ab}}]}\\
&=\sqrt{P^{opt (vN)}_S(\{\omega_i,\eta_i\})},
\end{align*}
where $\omega_i=\eta^{-1}_i\sqrt{\rho^{ab}}\ket{i}\bra{i}\otimes I^b\sqrt{\rho^{ab}}$ with $\eta_i=\mathrm{tr}\bra{i}\rho^{ab}\ket{i}$ and ${\{\pi_i\}}^{n_a}_{i=1}$ is a von Neumann measurement on $\mathcal{H}_a$. Moreover, the closest partial incoherent state (CPIS) of $\rho$ is given by
\begin{align}\label{eq27}
\sigma_{\rho}=\frac{1}{F^2(\rho^{ab})}\sum_i\ket{i}\bra{i}\otimes\bra{i}\sqrt{\rho^{ab}}\pi^{opt}_i\sqrt{\rho^{ab}}\ket{i},
\end{align}
where $\{\pi^{opt}_i\}^n_{i=1}$ is the optimal von Neumann measurement on system a.

We denote $\{\omega_i,\eta_i\}$ as the QSD-task of bipartite quantum state $\rho^{ab}$. Since
$P^{opt (vN)}_S(\{\omega_i,\eta_i\})\le P^{opt}_S(\{\omega_i,\eta_i\})$, we have
\begin{align*}
C^a_F(\rho^{ab})&=1-F^2(\rho^{ab})\\
&=1-P^{opt (vN)}_S(\{\omega_i,\eta_i\})\\
&\ge1-P^{opt}_S(\{\omega_i,\eta_i\})=P^{opt}_E(\{\omega_i,\eta_i\}^{n_a}_{i=1}).
\end{align*}

Next, we consider the ensemble $\{\omega_i,\eta_i\}^{n_a}_{i=1}$.

If $\eta_i\ne0$, for all $i$, this means the ensemble contains $n_a$ states. If $\{\omega_i\}^{n_a}_{i=1}$ are linearly independent, then the optimal measurement is von Neumann measurement \cite{Eldar2003}, that is,
  \begin{align}
  P^{opt}_E(\{\omega_i,\eta_i\}^{n_a}_{i=1})=P^{opt (vN)}_E(\{\omega_i,\eta_i\}^{n_a}_{i=1})=C^a_F(\rho^{ab}).\nonumber
  \end{align}

 If $s$ number of $\eta_i$ are zero then
 the ensemble is $\{\omega_{i^{\prime}},\eta_{i^{\prime}}\}^{n_a-s}_{i^{\prime}=1}$. As above, if $\{\omega_{i^{\prime}}\}^{n_a-s}_{i^{\prime}=1}$ are linearly independent,
 \begin{align*}
 C^a_F(\rho^{ab})&=1-\max_{\{\pi_i\}}\sum^{n_a}_{i=1}\eta_i\mathrm{tr}(\pi_i\omega_i)\\
 &=1-\max_{\{\pi_{i^{\prime}}\}}\sum^{n_a-s}_{i^{\prime}=1}\eta_{i^{\prime}}\mathrm{tr}(\pi_{i^{\prime}}\omega_{i^{\prime}})\\
 &=P^{opt (vN)}_E(\{\omega_{i^{\prime}},\eta_{i^{\prime}}\}^{n_a-s}_{i^{\prime}=1}),
 \end{align*}
 we have
 \begin{align*}
  C^a_F(\rho^{ab})=P^{opt}_E(\{\omega_{i^{\prime}},\eta_{i^{\prime}}\}^{n_a-s}_{i^{\prime}=1}).
 \end{align*}

\end{proof}

\begin{rem}
We have the QSD task $\omega_i=\eta^{-1}_i\sqrt{\rho^{ab}}\ket{i}\bra{i}\otimes I^b\sqrt{\rho^{ab}}$ with $\eta_i=\mathrm{tr}\bra{i}\rho^{ab}\ket{i}$.
If $\dim(\mathcal{H}_b)=1$ (that is, $\rho^{ab}$ reduces to a single quantum system $\rho^a$), then the corresponding ensemble constitutes a pure state discrimination task, which is consistent with coherence theory \cite{Xiong2018A}. In this sense, partial coherence theory is a more general framework to investigate QSD.
\end{rem}

\begin{cor}\label{cor11}(\cite{spehner2013A})
If $\rho>0$, the fidelity partial coherence is equal to the minimum error probability to discriminate $\{\omega_i,\eta_i\}^{n_a}_{i=1}$, that is,
\begin{align}\label{eq23}
C^a_F(\rho^{ab})=P^{opt}_E(\{\omega_i,\eta_i\}^{n_a}_{i=1}).
\end{align}
\end{cor}
\begin{proof}
If $\rho>0$, then each $\omega_i$ is full rank and $\eta_i>0$. Supposing every Hermitian matrix $R_i:=\bra{i}\rho^{ab}\ket{i}$ has spectral decomposition $R_i=\sum_j\lambda_{ij}\eta_i\ket{\xi_{ij}}\bra{\xi_{ij}}$ where $\ket{\xi_{ij}}\in\mathcal{H}_b$, it is easy to check that
\begin{align}
\ket{\zeta_{ij}}:=(\lambda_{ij}\eta_i)^{-1}\sqrt{\rho^{ab}}\ket{i}\otimes\ket{\xi_{ij}}
\end{align}
is an eigenvector of $\omega_i$ with eigenvalue $\lambda_{ij}>0$. Hence, for
\begin{align}
\sum_{ij}x_{ij}\ket{\zeta_{ij}}=\sqrt{\rho^{ab}} \left(\sum_{ij} \frac{x_{ij}}{\lambda_{ij}\eta_i}\ket{i}\otimes\ket{\xi_{ij}} \right)=0,
\end{align}
the invertibility of $\rho$ and orthogonality of $\{\ket{i}\}$, $\{\ket{\xi_{ij}}\}$ indicate that $x_{ij}=0$ for each $i,j$. As a result, $\{\rho_i,i=1,...,n_a\}$ are linearly independent and we obtain the result \eqref{eq23} using Theorem \ref{thm4}.
\end{proof}

\subsection{quantum state discrimination and partial coherence}

We show that estimating fidelity partial coherence of a bipartite quantum state can be regarded as a QSD task. In this connection we ask, whether for a given QSD ensemble, there exist a quantum state whose fidelity partial coherence provides an upper bound for the minimum error probability of QSD?

Let us consider a state discrimination task $\{\rho_i,\eta_i\}^n_{i=1}$ where each $\rho_i$ is an $m\times m$ density matrix. Then we consider an $mn\times mn$ matrix $\rho$ whose $(i,j)$-th entry is a block which reads $\rho_{ij}=\sqrt{\eta_i\rho_i}\sqrt{\eta_j\rho_j}$, $1\le i,j\le n$, that is,
 \begin{align}\label{eq21}
 \rho=\begin{pmatrix}
\eta_1\rho_1&\sqrt{\eta_1\rho_1}\sqrt{\eta_2\rho_2}&...& \sqrt{\eta_1\rho_1}\sqrt{\eta_n\rho_n}\\
\sqrt{\eta_2\rho_2}\sqrt{\eta_1\rho_1}&\eta_2\rho_2&...&\sqrt{\eta_2\rho_2}\sqrt{\eta_n\rho_n}\\
.&.&...&.\\.&.&...&.\\
\sqrt{\eta_n\rho_n}\sqrt{\eta_1\rho_1}&\sqrt{\eta_n\rho_n}\sqrt{\eta_2\rho_2}&...&\eta_n\rho_n
\end{pmatrix}.
\end{align}

\begin{prop}
The matrix $\rho$ is a density matrix.
\end{prop}
\begin{proof}
Consider an $m\times mn$ matrix
\begin{align}
A=(\sqrt{\eta_1\rho_1},\sqrt{\eta_2\rho_2},...,\sqrt{\eta_n\rho_n}),
\end{align}
 then $\rho=A^{\dagger}A$ is positive semidefinite. As $\mathrm{tr}\rho=\sum_i \mathrm{tr} (\eta_i \rho_i)=\sum_i \eta_i=1$, we conclude that $\rho$ is a density matrix.
\end{proof}

Therefore, we call state \eqref{eq21} the QSD-state of $\{\rho_i,\eta_i\}^n_{i=1}$.

Based on Theorem \ref{thm4}, the corresponding QSD ensemble of $\rho$ is $\{\omega_i,p_i\}^n_{i=1}$, where
\begin{align}
p_i=\mathrm{tr}\sqrt{\rho}\ket{i}\bra{i}\otimes I_m\sqrt{\rho},~\omega_i=p^{-1}_i\sqrt{\rho}\ket{i}\bra{i}\otimes I_m\sqrt{\rho}.\nonumber
\end{align}

Using polar decomposition theorem, one can find an $m\times mn$ unitary matrix $U$ such that $A=U\sqrt{\rho}$. As a result, for each $1\le i\le n$,
\begin{align*}
p_i&=\mathrm{tr}\sqrt{\rho}\ket{i}\bra{i}\otimes I_m\sqrt{\rho}\\
&=\mathrm{tr}U^{\dagger}(A\ket{i}\bra{i}\otimes I_mA^{\dagger})U\\
&=\mathrm{tr}(\eta_iU^{\dagger}\rho_iU)=\eta_i,
\end{align*}
and $\omega_i=U^{\dagger}\rho_iU$. Moreover,
\begin{align*}
P^{opt}_S(\{\rho_i,\eta_i\}^n_{i=1})&=\max_{\{M_i\}^n_{i=1}}\sum_i\eta_i\mathrm{tr}(M_i\rho_i)\\
&=\max_{\{M_i\}^n_{i=1}}\sum_i\eta_i\mathrm{tr}M_iU^{\dagger}\omega_iU\\
&=\max_{\{N_i\}^n_{i=1}}\sum_i\eta_i\mathrm{Tr}(N_i\omega_i)\\
&=P^{opt}_S(\{\omega_i,\eta_i\}^n_{i=1}).
\end{align*}
The last equality is due to the fact that if $\{M_i\}^n_{i=1}$ is a POVM on $\mathcal{H_A}$, then $\{UM_iU^{\dagger}\}^n_{i=1}$ is a POVM on $\mathcal{H}_A\otimes\mathcal{H}_B$. Therefore, $\{UM^{opt}_iU^{\dagger}\}^n_{i=1}$ is an optimal measurement for QSD task $\{\omega_i,\eta_i\}^n_{i=1}$ when $\{M^{opt}_i\}^n_{i=1}$ is optimal to discriminate $\{\rho_i,\eta_i\}^n_{i=1}$.

Thus, we have the following result.
\begin{thm}\label{thm5}
Let $\mathcal{H}_a$ and $\mathcal{H}_b$ are Hilbert spaces of systems a and b with $\dim(\mathcal{H}_a)=n$ and $\dim(\mathcal{H}_b)=m$. 
For a set of quantum states $\rho_i,i=1,...,n$ of system a, the minimal error probability to discriminate $\{\rho_i,\eta_i\}^n_{i=1}$ is upper bounded by the fidelity partial coherence with respect to $\{\ket{i}\bra{i}\otimes I_m,i=1,...,n\}$ of the corresponding QSD-state $\rho$, that is,
\begin{align}
P^{opt}_E(\{\rho_i,\eta_i\}^n_{i=1})\le C^a_F(\rho),
\end{align}
where $\{\ket{i}\}^n_{i=1}$ is the computational basis of $\mathcal{H}_a$ . The bound is saturated when these states are linearly independent.

\end{thm}
\begin{proof}

Using Theorem \ref{thm4}, we have
\begin{align}
C^{a}_F(\rho)=P^{opt (vN)}_E(\{\omega_i,p_i\}^n_{i=1}),\nonumber
\end{align}
where $\omega_i=\eta^{-1}_i\sqrt{\rho^{ab}}\ket{i}\bra{i}\otimes I_m\sqrt{\rho^{ab}}$ and $p_i=\eta_i$. Since
$\omega_i=U\rho_iU^{\dagger}$ for an unitary matrix $U$, then
\begin{align}
P^{opt}_E(\{\rho_i, \eta_i\}^n_{i=1})=P^{opt}_E(\{\omega_i, \eta_i\}^n_{i=1}).\nonumber
\end{align}
 Therefore, we have
 \begin{align}
 C^a_F(\rho)\ge P^{opt}_E(\{\rho_i, \eta_i\}^n_{i=1}),\nonumber
 \end{align}
 and the equality holds for linearly independent states $\{\ket{\rho_i}\}^n_{i=1}$.
\end{proof}

\subsection{fidelity partial coherence for $(2,n)$ bipartite X states}

In this section, we compute the analytic expression of fidelity partial coherence for X-states.
In two-qubit case, X-states including Bell-diagonal states constitute an important class of states which play a crucial role in the quantification and dynamics of entanglement, quantum correlations and coherence \cite{Vedral1997,luo2008A,Dakic2010,M.Ali2010A,Rau2009,Bromley2015}.

First, we consider the two-qubit case. The density matrix of a two-qubit X-state in the standard orthogonal basis $\{\ket{00},\ket{01},\ket{10},\ket{11}\}$ is of the general form
\begin{align}\label{eq37}
\rho_X=\begin{pmatrix}a&0&0&y\\0&b&x&0\\0&\overline{x}&c&0\\\overline{y}&0&0&d
\end{pmatrix}.
\end{align}
Therefore,
\begin{align}
\Lambda=\eta_1\rho_1-\eta_2\rho_2=\sqrt{\rho_X}\sigma_z\otimes I\sqrt{\rho_X},
\end{align}
where $\sigma_z=\ket{0}\bra{0}-\ket{1}\bra{1}$ is the Pauli matrix. $\Lambda$ has the same eigenvalues as $\sigma_z\otimes I\rho_X$,  whose eigenvalues and corresponding eigenvectors are
\begin{align}
&\lambda_{1(2)}=\frac{1}{2}(b-c\mp\sqrt{(b+c)^2-4|x|^2}), \\ &\lambda_{3(4)}=\frac{1}{2}(a-d\mp\sqrt{(a+d)^2-4|y|^2}),
\end{align}
and
\begin{align*}
&\ket{\psi_{1(2)}}=(0,(b+c)\mp\sqrt{(b+c)^2-4|x|^2},-2\overline{x},0)^{T},\\
&\ket{\psi_{3(4)}}=((a+d)\mp\sqrt{(a+d)^2-4|y|^2},0,0,-2\overline{y})^{T},
\end{align*}
respectively. Hence, we have
\begin{align*}
&P^{opt}_S(\{\omega_i,\eta_i\})=\frac{1}{2}(1+\mathrm{tr}|\Lambda|) \\
&=\frac{1}{2}(1+\sqrt{(b+c)^2-4|x|^2}+\sqrt{(a+d)^2-4|y|^2}).
\end{align*}

If $bc\neq|x|^2,ad\neq|y|^2$, that is, $\rho_X>0$, one has
\begin{align*}
C^a_F(\rho_X)&=P^{opt (vN)}_E(\{\omega_i,\eta_i\})\\
&=P^{opt}_E(\{\omega_i,\eta_i\})\\
&=\frac{1}{2}(1-\sqrt{(b+c)^2-4|x|^2}-\sqrt{(a+d)^2-4|y|^2}).
\end{align*}

And the closest partial incoherent state is given by Eq. $\eqref{eq27}$ with optimal projectors
\begin{align*}
\pi^{opt}_1&=\sqrt{\rho_X}(\ket{\psi_2}\bra{\psi_2}+\ket{\psi_4}\bra{\psi_4})\sqrt{\rho_X},\nonumber \\
\pi^{opt}_2&=I-\pi^{opt}_1,
\end{align*}
where $\ket{\psi_{2(4)}}$ are the normalized eigenvectors.

Nest, we consider $(2,n)$ bipartite quantum systems. Any $2n\times2n$ quantum state $\rho$ is called X-state if it can be represented as an X matrix in a fixed orthogonal basis  $\{\ket{i}\}^{2n}_{i=1}$ as
\begin{align}
\rho=\begin{pmatrix}
\rho_{11} &0&.&.&0&\rho_{1,2n}\\0&\rho_{22}&.&.&\rho_{2,2n-1}&0\\.&.&.&.&.&.\\
.&.&.&.&.&.\\0&\rho_{2n-1,1}&.&.&\rho_{2n-1,2n-1}&0\\ \rho_{2n,1} &0&.&.&0&\rho_{2n,2n}
\end{pmatrix}.
\end{align}

Similar to the $n=2$ case, the fidelity partial coherence of invertible $\rho$ is
\begin{align*}
C^a_F(\rho)=\frac{1}{2}(1-\sum^n_{i=1}\sqrt{(\rho_{ii}+\rho_{2n+1-i,2n+1-i})^2-4|\rho_{2n+1-i,i}|^2}).
\end{align*}

\section{affinity partial coherence and quantum state discrimination}\label{affinity-coherence}

In this section, firstly we evaluate affinity partial coherence for a bipartite quantum state $\rho^{ab} \in \mathcal{H}_a\otimes\mathcal{H}_b$. Since each partial incoherent state can be written as
\begin{align}
  \sigma=\sum_{ij}p_{ij}\ket{i}\bra{i}\otimes\ket{\psi_{j|i}}\bra{\psi_{j|i}},
\end{align}
affinity between $\rho^{ab}$ and partial incoherent states $\sigma \in I^a_P$ is given by
\begin{align*}
A(\rho^{ab}):&=\max_{\sigma\in I^a_P}A(\rho^{ab},\sigma)=\max_{\sigma\in I^a_P}\mathrm{tr}\sqrt{\rho^{ab}}\sqrt{\sigma}\\
&=\max_{\{\ket{\psi_{j|i}}\}}\sum_{ij}\sqrt{p_{i,j}}\bra{i\otimes\psi_{j|i}}\sqrt{\rho^{ab}}\ket{i\otimes\psi_{j|i}}\\
&\le\max_{\{\ket{\psi_{j|i}}\}}\sqrt{\sum_{ij}\bra{\psi_{j|i}}B_i\ket{\psi_{j|i}}^2}\\
&=\sqrt{\sum_{i}\mathrm{tr}B^2_i},
\end{align*}
where $B_i=\bra{i}\otimes I_b\sqrt{\rho^{ab}}\ket{i}\otimes I_b$. While the inequality is due to the Cauchy-Schwarz inequality, the last equality holds when $\ket{\psi_{j|i}}$ is the eigenvector of $B_i$ for each $i$. As a result,
\begin{align}
C^a_A(\rho^{ab})=1-\sum_{i}\mathrm{tr}(\bra{i}\otimes I_b\sqrt{\rho^{ab}}\ket{i}\otimes I_b)^2,
\end{align}
and the corresponding CPIS is
\begin{align}
\sigma=\sum_{i,j}q_{ij}\ket{i}\bra{i}\otimes\ket{\beta_{j|i}}\bra{\beta_{j|i}}
\end{align}
with
\begin{align*}
q_{ij}=\frac{\bra{i\otimes\beta_{j|i}}\sqrt{\rho^{ab}}\ket{i\otimes\beta_{j|i}}^2}{\sum_{i,j}\bra{i\otimes\beta_{j|i}}\sqrt{\rho^{ab}}\ket{i\otimes\beta_{j|i}}^2}
\end{align*}
and $\ket{\beta_{j|i}}$ is the optimal eigenvector of $B_i$ for each $i$.
\begin{thm}
$C^a_A$ is a partial coherence measure.
\end{thm}

\begin{proof}
Earlier, at the end of Sec. \ref{fidelity-coherence} A, we have seen that $C^a_A$ is a strong partial coherence monotone. Here, we just need to prove the convexity of $C^a_A$. Let us consider the quantification of partial coherence using Wigner-Yanase skew information \cite{Luo2017B}
\begin{align}
C^a_H(\rho^{ab}):=\sum_iI(\rho,\ket{i}\bra{i}\otimes I),
\end{align}
where $I(\rho,K):=-\frac{1}{2}\mathrm{tr}[\sqrt{\rho},K]^2$. Since
\begin{align*}
C^a_H(\rho^{ab})&=\sum_i\mathrm{tr}(\rho\ket{i}\bra{i}\otimes I-\sqrt{\rho}\ket{i}\bra{i}\otimes I\sqrt{\rho}\ket{i}\bra{i}\otimes I)\\
&=1-\sum_i\mathrm{tr}(\sqrt{\rho}\ket{i}\bra{i}\otimes I\sqrt{\rho}\ket{i}\bra{i}\otimes I)\\
&=C^a_A(\rho^{ab}),
\end{align*}
and $I(\rho,K)$ is convex in $\rho$ \cite{wigner1963,Luo2003}, one has
\begin{align*}
&C^a_A \left(\sum_ip_i\rho^{ab}_i \right)=C^a_H \left(\sum_ip_i\rho^{ab}_i \right)\\
\le&\sum_ip_iC^a_H(\rho^{ab}_i)=\sum_ip_iC^a_A(\rho^{ab}_i).
\end{align*}
\end{proof}

Now we consider the quantum state discrimination of the ensemble $\{\omega_i,\eta_i\}^{n_a}_{i=1}$. As $\sum_i\eta_i\omega_i
=\rho^{ab}$, the LSM are given by
\begin{align*}
M^{lsm}_i=\eta_i(\rho^{ab})^{-1/2}\omega_i(\rho^{ab})^{-1/2}=\ket{i}\bra{i}\otimes I^b, i=1,...,n,
\end{align*}
where $\rho^{-1/2}:=\sum_i\lambda^{-1/2}_iP_i$ for the spectral decomposition $\rho=\sum_i\lambda_iP_i$. Hence, the success probability to discriminate $\{\omega_i,\eta_i\}$ with LSM is
\begin{align*}
P^{lsm}_S(\{\omega_i,\eta_i\})&=\sum_i\eta_i\mathrm{tr}(M^{lsm}_i\omega_i)\\
&=\sum_i\bra{i}\otimes I\sqrt{\rho^{ab}}\ket{i}\bra{i}\otimes I\sqrt{\rho^{ab}}\ket{i}\otimes I\\
&=\sum_iB^2_i=A^2(\rho^{ab}).
\end{align*}
Thus, we have the following theorem.

\begin{thm}
For a bipartite quantum state $\rho^{ab} \in \mathcal{H}_a\otimes\mathcal{H}_b$
with $\{\ket{i}\}^n_{i=1}$ being a reference basis of $\mathcal{H}_a$, the affinity partial coherence of $\rho^{ab}$ is equal to the error probability to discriminate $\{\omega_i,\eta_i\}^n_{i=1}$ with least square measurement. That is,
\begin{align}
C^a_A(\rho^{ab})=P^{lsm}_E(\{\omega_i,\eta_i\}^n_{i=1}),
\end{align}
where $\eta_i=\mathrm{tr}\bra{i}\rho^{ab}\ket{i}$ and $\omega_i=\eta^{-1}_i\sqrt{\rho^{ab}}\ket{i}\bra{i}\otimes I^b\sqrt{\rho^{ab}}$.
\end{thm}

On the other hand, for a state discrimination task $\{\rho_i,\eta_i\}^n_{i=1}$ with QSD-state $\rho$, one can find an unitary matrix $U$ such that
\begin{align*}
\rho_i=U\omega_i U^{\dagger}, ~i=1,...,n,
\end{align*}
where $\eta_i,\omega_i$ are the same as above. Therefore, the LSM for $\{\rho_i,\eta_i\}^n_{i=1}$ are
\begin{align*}
N_i=U^{\dagger}\ket{i}\bra{i}\otimes IU,
\end{align*}
and the error probability to discriminate this task with LSM is
\begin{align*}
P^{lsm}_E(\{\rho_i,\eta_i\})&=1-\sum_i\eta_i\mathrm{tr}(U^{\dagger}\ket{i}\bra{i}\otimes IU\rho_i)\\
&=1-\sum_i\eta_i\mathrm{tr}(\ket{i}\bra{i}\otimes I\omega_i)\\
&=P^{lsm}_E(\{\omega_i,\eta_i\})=C^a_A(\rho).
\end{align*}

Hence, we arrive at the following result.
\begin{thm}
Let $\mathcal{H}_a$ be an $n$-dimensional Hilbert space and $\{\ket{i}\}^n_{i=1}$ be the computational basis, that is, the $i$-th entry of vector $\ket{i}$ is $1$ and rest are zero. For the ensemble $\{\rho_i,\eta_i\}^n_{i=1}$ in the Hilbert space $\mathcal{H}_b$, the error probability to discriminate $\{\rho_i\}^n_{i=1}$ with LSM is equal to the affinity partial coherence of the corresponding QSD-state $\rho$, that is,
\begin{align}
P^{lsm}_E(\{\rho_i,\eta_i\}^n_{i=1})= C^a_A(\rho),
\end{align}
where the fixed L{\"u}ders measurement are $\{\ket{i}\bra{i}\otimes I^b\}^n_{i=1}$.
\end{thm}

\section{correlated coherence and quantum correlation}
\label{sec:correlated coherence}

Let $X$ be either fidelity or affinity below. We can define quantum entanglement, quantum discord, coherence and partial coherence from a unified perspective as follows:
\begin{align*}
&E_X(\rho^{ab}):=\min_{\sigma\in\mathcal{S}}d_X(\rho^{ab},\sigma),\\
&D_X(\rho^{ab}):=\min_{\sigma\in \mathcal{C}_a}d_X(\rho^{ab},\sigma),\\
&C_X(\rho):=\min_{\sigma\in \mathcal{I}}d_X(\rho,\sigma),\\
&C^a_X(\rho^{ab}):=\min_{\sigma\in I^a_P}d_X(\rho^{ab},\sigma),
\end{align*}
where  $\mathcal{S}$, $\mathcal{I}$ and $\mathcal{C}_a$ respectively are the sets of separable states, incoherent states and classical states on party $a$ \cite{spehner2013A}. Using Theorem \ref{thm2} and Remark \ref{rem4}, $E_A$ is a strong entanglement monotone and $E_F$ (geometric entanglement) is a entanglement measure \cite{Wei2003}.

We know that coherence is a more fundamental quantum correlation than entanglement and discord \cite{Yao2015}. Moreover, partial coherence in a bipartite system may contain both local coherence and correlated coherence. To characterize the correlation between parties a and b, we define generalized correlated coherence, following Refs. \cite{Tan2016,Tan2018}, as
\begin{align}
C^a_{X,gcc}(\rho^{ab}):=C^a_X(\rho^{ab})-C_X(\rho^a).
\end{align}

\begin{thm}
The generalized correlated coherence $C^a_{X,gcc}$ is nonnegative.
\end{thm}

\begin{proof}
Since
\begin{align*}
C^a_X(\rho^{ab})
&=\min_{\{p_i,\sigma_i\}}d_X \left(\rho^{ab},\sum_ip_i\ket{i}\bra{i}\otimes\sigma_i \right)\\
&\ge\min_{\{p_i,\sigma_i\}}d_X \left(\mathrm{tr}_b\rho^{ab},\mathrm{tr}_b \left(\sum_ip_i\ket{i}\bra{i}\otimes\sigma_i \right)\right)\\
&=\min_{\{p_i\}}d_X \left(\rho^a,\sum_ip_i\ket{i}\bra{i} \right)\\
&=C_X(\rho^{ab}),
\end{align*}
where the inequality is due to the contractibility of ``$d_X$", the generalized correlated coherence $C^a_{X,gcc}(\rho^{ab})$ is nonnegative.
\end{proof}

Our definition of correlated coherence is basis-dependent. However, we can also define basis-independent correlated coherence in a natural way. For bipartite state $\rho^{ab}$, the reduced density matrix $\rho^{a}$ has eigenvectors $\{\ket{\alpha_i}\}$. Choosing the local basis $\rho^{a}$ has zero (local) coherence, and the correlated coherence reduces to the partial coherence with respect to the eigenbasis of the corresponding reduced density matrix. In this way, the partial coherence in the system is stored entirely within the correlations.


Therefore, we define the correlated coherence as
\begin{align}
C^a_{X,cc}:=\min_{\mathcal{B}_a}C^a_X(\rho^{ab})-C_X(\rho^a),
\end{align}
where the minimization is performed over all the local bases $\mathcal{B}_a$ satisfying $C_X(\rho^a)=0$.

\begin{thm}\label{thm11}
For a bipartite quantum state $\rho^{ab}$, 
\begin{align}
C^a_{X,cc}(\rho^{ab})\ge D^a_X(\rho^{ab}).
\end{align}
The equality holds if either $\rho^a$ is a completely mixed state or $\rho^{ab}$ is a pure state.
\end{thm}
\begin{proof}
As $D^a_X(\rho^{ab})$ is the minimal partial coherence of $\rho^{ab}$ \cite{Luo2017B}, we have $C^a_{X,cc}(\rho^{ab})\ge D^a_X(\rho^{ab})$.

Now, if $\rho^a=\frac{1}{n_a}I^a$, then the eigenbasis of $\rho^a$ can be any set of orthogonal basis in $H_a$. As a result,
\begin{align*}
C^a_X(\rho^{ab})&=\min_{\mathcal{B}_a}C^a_{X,cc}(\rho^{ab})\\
&=\min_{\{p_i,\ket{\alpha_i},\sigma_i\}}d_X \left(\rho^{ab},\sum_ip_i\ket{\alpha_i}\bra{\alpha_i}\otimes\sigma_i \right)\\
&=D^a_X(\rho^{ab}).
\end{align*}

And, for a pure bipartite state $\ket{\psi}$ with the Schmidt decomposition
\begin{align}\label{eq38}
\ket{\psi}=\sum_i\sqrt{\lambda_i}\ket{x_i}\otimes\ket{y_i},
\end{align}
$\rho_{\psi}^a:=\mathrm{tr}_b\ket{\psi}\bra{\psi}=\sum_i\lambda_i\ket{x_i}\bra{x_i}$. Note that when some Schmidt coefficients are equal, the Schmidt decomposition in Eq. \eqref{eq38} is not unique.
Without loss of generalization, we assume that $\lambda_1\ge\lambda_2\ge...\ge\lambda_n$ and the Schmidt decomposition in Eq. \eqref{eq38} satisfies $C^a_{A,cc}(\rho^{ab})=C^a_A(\rho^{ab})$ (with respect to $\{\ket{x_i}\}$). Then,
\begin{align*}
\max_{\sigma\in I^a_P}A(\ket{\psi},\sigma)&=\max_{\{p_i,\sigma_i\}}\sum_i\sqrt{p_i}\bra{\psi} (\ket{x_i}\bra{x_i}\otimes\sigma_i) \ket{\psi}\\
&=\max_{\{p_i,\sigma_i\}}\sum_i\sqrt{p_i}\lambda_i\bra{y_i}\sigma_i\ket{y_i}\\
&=\max_{\{p_i\}}\sum_i\sqrt{p_i}\lambda_i\le\sqrt{\sum_i\lambda^2_i}.
\end{align*}
The third equality holds when we choose each $\sigma_i=\ket{y_i}\bra{y_i}$. The $``\le"$ is due to the Cauchy-Schwartz inequality and the maximum is reached for $p_i=\frac{\lambda^2_i}{\sum_j\lambda^2_j}$. Hence,
\begin{align*}\label{eq39}
C^a_{A,cc}(\ket{\psi})=1-\sum_i\lambda^2_i,
\end{align*}
with $\lambda_i$-s being the Schmidt coefficients of $\ket{\psi}$ and the closest a-classical state is
\begin{align}
\sigma=\sum_i\frac{\lambda^2_i}{\sum_j\lambda^2_j}\ket{x_i\otimes y_i}\bra{x_i\otimes y_i}.
\end{align}


On the other hand, if we assume that the Schmidt decomposition in Eq. \eqref{eq38} satisfies $C^a_{F,cc}(\rho^{ab})=C^a_F(\rho^{ab})$ (with respect to $\{\ket{x_i}\}$), then
\begin{align*}
\max_{\sigma\in I^a_P}F(\ket{\psi},\sigma)&=\max_{\{p_i,\sigma_i\}}\sqrt{\sum_ip_i\lambda_i\bra{y_i}\sigma_i\ket{y_i}}\\
&=\max_{\{p_i\}}\sqrt{\sum_ip_i\lambda_i}\le\sqrt{\lambda_1},
\end{align*}
and 
\begin{align*}
C^a_{F,cc}(\ket{\psi})=1-\lambda_{max}.
\end{align*}
If the maximal Schmidt coefficient is unity, the closest a-classical state is a pure state $\ket{x_1\otimes y_1}$. If $\lambda_1=...=\lambda_m$, then the closest a-classical states are infinite, say
\begin{align}
\sigma_{\psi}=\sum^m_{i=1}p_i\ket{x_i\otimes y_i}\bra{x_i\otimes y_i},
\end{align}
with $p_i\in[0,1]$ and $\sum_ip_i=1$.
\end{proof}

\begin{thm}
$C^a_{X,cc}$ is a measure of quantum correlation for a bipartite quantum state $\rho^{ab}$. That is,

(1) $C^a_{X,cc}(\rho^{ab})\ge0$; ``=" holds if and only if $\rho^{ab}=\sum_ip_i\ket{\alpha_i}\bra{\alpha_i}\otimes\sigma_i$.

(2) $C^a_{X,cc}(\rho^{ab})$ is invariant under local unitary transformation.

(3) $C^a_{X,cc}$ is monotonically non-increasing under quantum operations on b, i.e., $C^a_{cc}(I^a\otimes \Phi^b(\rho^{ab}))\le C^a_{cc}(\rho^{ab})$ for any quantum operation $\Phi^b$.

(4) $C^a_{X,cc}$ reduces to an entanglement monotone for pure states.
\end{thm}

\begin{proof}
(1) Nonnegativity is obvious, so we just need to prove the second part. $C^a_{X,cc}(\rho^{ab})=0$ indicates that $\rho^{ab}=\sum_ip_i\ket{\alpha_i}\bra{\alpha_i}\otimes\sigma_i$ with $\{\ket{\alpha_i}\}$ being the eigenbasis of $\rho^a$.
On the other hand, if $\rho^{ab}=\sum_ip_i\ket{\alpha_i}\bra{\alpha_i}\otimes\sigma_i$, then $\rho^{a}=\sum_ip_i\ket{\alpha_i}\bra{\alpha_i}$. As a result, the eigenbasis of $\rho^a$ is $\{\ket{\alpha_i}\}$ and $C^a_{X,cc}(\rho^{ab})=0$.

(2) Suppose the referred eigenbasis of $\rho^a$ is $\{\ket{x_i}\}$. As $\mathrm{tr}_b(U_a\otimes U_b\rho^{ab}U^{\dagger}_a\otimes U^{\dagger}_b)=U_a\rho^{a}U^{\dagger}_a$, the corresponding eigenbasis is $\{U_a\ket{x_i}\}$. Hence,
\begin{align*}
&C^a_{X,cc}(U_a\otimes U_b\rho^{ab}U^{\dagger}_a\otimes U^{\dagger}_b)\\
=&\min_{\{p_i,\sigma_i\}}d_X \left(U_a\otimes U_b\rho^{ab}U^{\dagger}_a\otimes U^{\dagger}_b,\sum_ip_iU_a\ket{x_i}\bra{x_i}U^{\dagger}_a\otimes\sigma_i \right)\\
=&\min_{\{p_i,\sigma_i\}}d_X \left(\rho^{ab},\sum_ip_i\ket{x_i}\bra{x_i}\otimes\sigma_i \right)=C^a_{X,cc}(\rho^{ab}).
\end{align*}

(3) Assuming that $\sigma^{\star}\in \mathcal{C}_a$ is the CPIS of $\rho^{ab}$, we have
\begin{align*}
C^a_{X,cc}(\rho^{ab})&=d_X(\rho^{ab},\sigma^{\star})\\
&\ge d_X \left(I^a\otimes \Phi^b(\rho^{ab}),I^a\otimes \Phi^b(\sigma^{\star}) \right)\\
&\ge C^a_{X,cc} \left(I^a\otimes \Phi^b(\rho^{ab}) \right),
\end{align*}
where the first $``\ge"$ is due to the contractibility of $d_X$ and the second $``\ge"$ follows from $I^a\otimes \Phi^b(\sigma^{\star})\in \mathcal{C}_a$. In fact, if we suppose that $\sigma^{\star}=\sum_ip_i\ket{\alpha_i}\bra{\alpha_i}\otimes\sigma_i$, then $I^a\otimes \Phi^b(\sigma^{\star})=\sum_ip_i\ket{\alpha_i}\bra{\alpha_i}\otimes\Phi^b(\sigma_i)$ is also an a-classical state.

(4) 
Let $\lambda_i$s and $\mu_i$s be the Schmidt coefficients of bipartite pure states $\ket{\psi}$ and $\ket{\phi}$ respectively, where $\lambda_1\ge\lambda_2\ge...\ge\lambda_n$ and $\mu_1\ge\mu_2\ge...\ge\mu_n$. Denote $\vec{\lambda}=(\lambda_1,...,\lambda_n)^T$ and $\vec{\mu}=(\mu_1,...,\mu_n)^T$. If there is an LOCC channel which maps $\ket{\psi}$ to $\ket{\phi}$ then $\vec{\lambda}\prec\vec{\mu}$ \cite{Nielsen1999}, and there exists a doubly stochastic matrix $A=(a_{ij})_{n\times n}$ such that $\vec{\lambda}=A\vec{\mu}$. Hence,
\begin{align*}
\sum_i\lambda^2_i&=\sum_i \left(\sum_ja_{ij}\mu_j \right)^2\\
&=\sum_i \left(\sum_ja^2_{ij}\mu^2_j+2\sum_{j<k}a_{ij}a_{ik}\mu_j\mu_k \right)\\
&\le\sum_i \left[\sum_ja^2_{ij}\mu^2_j+\sum_{j<k}a_{ij}a_{ik} \left(\mu^2_j+\mu^2_k \right) \right]\\
&=\sum_i \left(\sum_ja^2_{ij}\mu^2_j+\sum_{j\ne k}a_{ij}a_{ik}\mu^2_j \right)\\
&=\sum_j\mu^2_j \left(\sum_ia^2_{ij}+\sum_i\sum_{k\ne j}a_{ij}a_{ik} \right)\\
&=\sum_j\mu^2_j \left[\sum_ia_{ij} \left(a_{ij}+\sum_{k\ne j}a_{ik} \right) \right]=\sum_j\mu^2_j.
\end{align*}
Moreover, based on Eq. \eqref{eq39}, one has
$C^a_{A,cc}(\ket{\psi})\ge C^a_{A,cc}(\ket{\phi}),$
which means that $C^a_{A,cc}$ is also an entanglement monotone.
Also, since $C^a_{F,cc}(\ket{\psi})=E_F(\ket{\psi})$, we conclude that $C^a_{X,cc}$ reduces to an entanglement monotone for pure states.

%
\end{proof}



\section{conclusion}
\label{sec:conclusion}

In summary, we prove that fidelity distance and affinity distance satisfy the strong contractibility condition. Moreover, under two assumptions, namely, convexity of free states and closure of free states under selective free operations, we show that resource quantifiers based on these distances are valid resource measures for a generic resource theory including entanglement, coherence, partial coherence and superposition, providing thereby a unified framework for different quantum resources.

Next, we employ these two resource quantifiers to partial coherence theory. By linking them to quantum mixed state discrimination task, we offer operational interpretation for these two partial coherence measures. Our results thus establish a useful connection between partial coherence and quantum mixed state discrimination task.

We also study correlated coherence under the framework of partial coherence theory. We show that correlated coherence is a kind of quantum discord. Our result, thus, reveals an interesting relation between partial coherence and quantum discord.

\begin{acknowledgments}
We thank Bartosz Regula for useful comments, and Mark Wilde
for telling us about several references on the concept of ``affinity''.
This project is supported by the National Natural Science Foundation of China (Grants No. 11171301, No. 11571307 and No. 61877054). SD and US thank the Zhejiang University for hospitality.
\end{acknowledgments}


\begin{thebibliography}{63}%
\makeatletter
\providecommand \@ifxundefined [1]{%
 \@ifx{#1\undefined}
}%
\providecommand \@ifnum [1]{%
 \ifnum #1\expandafter \@firstoftwo
 \else \expandafter \@secondoftwo
 \fi
}%
\providecommand \@ifx [1]{%
 \ifx #1\expandafter \@firstoftwo
 \else \expandafter \@secondoftwo
 \fi
}%
\providecommand \natexlab [1]{#1}%
\providecommand \enquote  [1]{``#1''}%
\providecommand \bibnamefont  [1]{#1}%
\providecommand \bibfnamefont [1]{#1}%
\providecommand \citenamefont [1]{#1}%
\providecommand \href@noop [0]{\@secondoftwo}%
\providecommand \href [0]{\begingroup \@sanitize@url \@href}%
\providecommand \@href[1]{\@@startlink{#1}\@@href}%
\providecommand \@@href[1]{\endgroup#1\@@endlink}%
\providecommand \@sanitize@url [0]{\catcode `\\12\catcode `\$12\catcode
  `\&12\catcode `\#12\catcode `\^12\catcode `\_12\catcode `\%12\relax}%
\providecommand \@@startlink[1]{}%
\providecommand \@@endlink[0]{}%
\providecommand \url  [0]{\begingroup\@sanitize@url \@url }%
\providecommand \@url [1]{\endgroup\@href {#1}{\urlprefix }}%
\providecommand \urlprefix  [0]{URL }%
\providecommand \Eprint [0]{\href }%
\providecommand \doibase [0]{http://dx.doi.org/}%
\providecommand \selectlanguage [0]{\@gobble}%
\providecommand \bibinfo  [0]{\@secondoftwo}%
\providecommand \bibfield  [0]{\@secondoftwo}%
\providecommand \translation [1]{[#1]}%
\providecommand \BibitemOpen [0]{}%
\providecommand \bibitemStop [0]{}%
\providecommand \bibitemNoStop [0]{.\EOS\space}%
\providecommand \EOS [0]{\spacefactor3000\relax}%
\providecommand \BibitemShut  [1]{\csname bibitem#1\endcsname}%
\let\auto@bib@innerbib\@empty

\bibitem{Horodecki2009R} R. Horodecki, P. Horodecki, M. Horodecki, and K. Horodecki, 
\href {https://link.aps.org/doi/10.1103/RevModPhys.81.865} {Rev. Mod. Phys. {\bf 81}, 865 (2009)}, {\it and references therein};
O. G{\"u}hne and G. Toth, 
\href {https://www.sciencedirect.com/science/article/abs/pii/S0370157309000623} {Phys. Rep. \textbf{474}, 1 (2009)};
S. Das, T. Chanda, M. Lewenstein, A. Sanpera, A. Sen(De), and U. Sen,
\href {https://arxiv.org/abs/1701.02187} {arXiv:1701.02187}.

\bibitem{Modi2012} K. Modi, A. Brodutch, H. Cable, T. Paterek, and V. Vedral,
\href {https://link.aps.org/doi/10.1103/RevModPhys.84.1655} {Rev. Mod. Phys. {\bf 84}, 1655 (2012)}, {\it and references therein};
A. Bera, T. Das, D. Sadhukhan, S. Singha  Roy, A. Sen(De), and U. Sen, 
\href {iopscience.iop.org/article/10.1088/1361-6633/aa872f/pdf} {Rep. Prog. Phys. \textbf{81}, 024001 (2018)}.

\bibitem {nielsen10} M. A. Nielsen and I. L. Chuang,
\href {http://dx.doi.org/10.1017/CBO9780511976667}
{\it Quantum Computation and Quantum Information}, (Cambridge University Press, 2010).

\bibitem{Aberg2006} J. {\AA}berg, \href {http://arxiv.org/abs/quant-ph/0612146}
  {arXiv:quant-ph/0612146}.

\bibitem [{\citenamefont {Baumgratz}\ \emph {et~al.}(2014)\citenamefont
  {Baumgratz}, \citenamefont {Cramer},\ and\ \citenamefont
  {Plenio}}]{Baumgratz2014}%
  \BibitemOpen
  \bibfield  {author} {\bibinfo {author} {\bibfnamefont {T.}~\bibnamefont
  {Baumgratz}}, \bibinfo {author} {\bibfnamefont {M.}~\bibnamefont {Cramer}}, \
  and\ \bibinfo {author} {\bibfnamefont {M.~B.}\ \bibnamefont {Plenio}},\
  }\href {\doibase 10.1103/PhysRevLett.113.140401} {\bibfield  {journal}
  {\bibinfo  {journal} {Phys. Rev. Lett.}\ }\textbf {\bibinfo {volume} {113}},\
  \bibinfo {pages} {140401} (\bibinfo {year} {2014})}\BibitemShut {NoStop}%
%
\bibitem [{\citenamefont {Winter}\ and\ \citenamefont
  {Yang}(2016)}]{Winter2016}%
  \BibitemOpen
  \bibfield  {author} {\bibinfo {author} {\bibfnamefont {A.}~\bibnamefont
  {Winter}}\ and\ \bibinfo {author} {\bibfnamefont {D.}~\bibnamefont {Yang}},\
  }\href {\doibase 10.1103/PhysRevLett.116.120404} {\bibfield  {journal}
  {\bibinfo  {journal} {Phys. Rev. Lett.}\ }\textbf {\bibinfo {volume} {116}},\
  \bibinfo {pages} {120404} (\bibinfo {year} {2016})}\BibitemShut {NoStop}%
  %
\bibitem [{\citenamefont {Streltsov}\ \emph {et~al.}(2017)\citenamefont
  {Streltsov}, \citenamefont {Adesso},\ and\ \citenamefont
  {Plenio}}]{streltsov2017A}%
  \BibitemOpen
  \bibfield  {author} {\bibinfo {author} {\bibfnamefont {A.}~\bibnamefont
  {Streltsov}}, \bibinfo {author} {\bibfnamefont {G.}~\bibnamefont {Adesso}}, \
  and\ \bibinfo {author} {\bibfnamefont {M.~B.}\ \bibnamefont {Plenio}},\
  }\href {\doibase 10.1103/RevModPhys.89.041003} {\bibfield  {journal}
  {\bibinfo  {journal} {Rev. Mod. Phys.}\ }\textbf {\bibinfo {volume} {89}},\
  \bibinfo {pages} {041003} (\bibinfo {year} {2017})}\BibitemShut {NoStop}%
  %
\bibitem [{\citenamefont {Marvian}\ and\ \citenamefont
  {Spekkens}(2013)}]{Marvian2013}%
  \BibitemOpen
  \bibfield  {author} {\bibinfo {author} {\bibfnamefont {I.}~\bibnamefont
  {Marvian}}\ and\ \bibinfo {author} {\bibfnamefont {R.~W.}\ \bibnamefont
  {Spekkens}},\ }\href {http://stacks.iop.org/1367-2630/15/i=3/a=033001}
  {\bibfield  {journal} {\bibinfo  {journal} {New J. Phys.}\ }\textbf
  {\bibinfo {volume} {15}},\ \bibinfo {pages} {033001} (\bibinfo {year}
  {2013})}\BibitemShut {NoStop}%
  %
\bibitem [{\citenamefont {Marvian}\ and\ \citenamefont
  {Spekkens}(2014)}]{Marvian2014}%
  \BibitemOpen
  \bibfield  {author} {\bibinfo {author} {\bibfnamefont {I.}~\bibnamefont
  {Marvian}}\ and\ \bibinfo {author} {\bibfnamefont {R.~W.}\ \bibnamefont
  {Spekkens}},\ }\href {http://dx.doi.org/10.1038/ncomms4821} {\bibfield
  {journal} {\bibinfo  {journal} {Nat. Comm.}\ }\textbf {\bibinfo
  {volume} {5}},\ \bibinfo {pages} {3821} (\bibinfo {year} {2014})}\BibitemShut
  {NoStop}%
  %
\bibitem [{\citenamefont {Brand\~ao}\ \emph {et~al.}(2015)\citenamefont
  {Brand\~ao}, \citenamefont {Horodecki}, \citenamefont {Ng}, \citenamefont
  {Oppenheim},\ and\ \citenamefont {Wehner}}]{Brandao2015}%
  \BibitemOpen
  \bibfield  {author} {\bibinfo {author} {\bibfnamefont {F.}~\bibnamefont
  {Brand\~ao}}, \bibinfo {author} {\bibfnamefont {M.}~\bibnamefont
  {Horodecki}}, \bibinfo {author} {\bibfnamefont {N.}~\bibnamefont {Ng}},
  \bibinfo {author} {\bibfnamefont {J.}~\bibnamefont {Oppenheim}}, \ and\
  \bibinfo {author} {\bibfnamefont {S.}~\bibnamefont {Wehner}},\ }\href
  {\doibase 10.1073/pnas.1411728112} {\bibfield  {journal} {\bibinfo  {journal}
  {Proc. Nat. Acad. Sci.}\ }\textbf {\bibinfo
  {volume} {112}},\ \bibinfo {pages} {3275} (\bibinfo {year}
  {2015})}\BibitemShut {NoStop}%
  %
\bibitem [{\citenamefont {Horodecki}\ and\ \citenamefont
  {Oppenheim}(2013)}]{Horodecki2013}%
  \BibitemOpen
  \bibfield  {author} {\bibinfo {author} {\bibfnamefont {M.}~\bibnamefont
  {Horodecki}}\ and\ \bibinfo {author} {\bibfnamefont {J.}~\bibnamefont
  {Oppenheim}},\ }\href {http://dx.doi.org/10.1038/ncomms3059} {\bibfield
  {journal} {\bibinfo  {journal} {Nat. Comm.}\ }\textbf {\bibinfo
  {volume} {4}},\ \bibinfo {pages} {2059} (\bibinfo {year} {2013})}\BibitemShut
  {NoStop}%
  %
\bibitem{Theurer2017} T. Theurer, N. Killoran, D. Egloff, and M. B. Plenio, 
\href {https://link.aps.org/doi/10.1103/PhysRevLett.119.230401} {Phys. Rev. Lett. {\bf 119}, 230401 (2017)};
S. Das, C. Mukhopadhyay, S. Singha Roy, S. Bhattacharya, A. Sen(De), and U. Sen,
\href {https://arxiv.org/abs/1705.04343} {arXiv:1705.04343}.
\bibitem{Wei2003} A. Shimony, 
\href {https://onlinelibrary.wiley.com/doi/abs/10.1111/j.1749-6632.1995.tb39008.x} {Ann. N. Y. Acad. Sci. \textbf{755}, 675 (1995)};
T.-C. Wei and P. M. Goldbart, 
\href {link.aps.org/pdf/10.1103/PhysRevA.68.042307} {Phys. Rev. A {\bf 68}, 042307 (2003)};
A. Sen(De) and U. Sen, 
\href {https://link.aps.org/doi/10.1103/PhysRevA.81.012308} {Phys. Rev. A \textbf{81}, 012308 (2010)};
A. Biswas, R. Prabhu, A. Sen(De), and U. Sen,
\href {https://link.aps.org/doi/10.1103/PhysRevA.90.032301} {Phys. Rev. A \textbf{90}, 032301 (2014)}.
\bibitem [{\citenamefont {Regula}(2018)}]{Regula2018A}%
  \BibitemOpen
  \bibfield  {author} {\bibinfo {author} {\bibfnamefont {B.}~\bibnamefont
  {Regula}},\ }\href {http://stacks.iop.org/1751-8121/51/i=4/a=045303}
  {\bibfield  {journal} {\bibinfo  {journal} {J. Phys. A:
  Math. Theor.}\ }\textbf {\bibinfo {volume} {51}},\ \bibinfo
  {pages} {045303} (\bibinfo {year} {2018})}\BibitemShut {NoStop}%
\bibitem [{\citenamefont {Luo}\ and\ \citenamefont {Sun}(2017)}]{Luo2017A}%
  \BibitemOpen
  \bibfield  {author} {\bibinfo {author} {\bibfnamefont {S.}~\bibnamefont
  {Luo}}\ and\ \bibinfo {author} {\bibfnamefont {Y.}~\bibnamefont {Sun}},\
  }\href {\doibase 10.1103/PhysRevA.96.022130} {\bibfield  {journal} {\bibinfo
  {journal} {Phys. Rev. A}\ }\textbf {\bibinfo {volume} {96}},\ \bibinfo
  {pages} {022130} (\bibinfo {year} {2017})}\BibitemShut {NoStop}%
  \bibitem [{\citenamefont {Kim}\ \emph {et~al.}(2018)\citenamefont {Kim},
\citenamefont {Li}, \citenamefont {Kumar},\ and\ \citenamefont
  {Wu}}]{Kim2018B}%
  \BibitemOpen
  \bibfield  {author} {\bibinfo {author} {\bibfnamefont {S.}~\bibnamefont
  {Kim}}, \bibinfo {author} {\bibfnamefont {L.}~\bibnamefont {Li}}, \bibinfo
  {author} {\bibfnamefont {A.}~\bibnamefont {Kumar}}, \ and\ \bibinfo {author}
  {\bibfnamefont {J.}~\bibnamefont {Wu}},\ }\href {\doibase
  10.1103/PhysRevA.98.022306} {\bibfield  {journal} {\bibinfo  {journal} {Phys.
  Rev. A}\ }\textbf {\bibinfo {volume} {98}},\ \bibinfo {pages} {022306}
  (\bibinfo {year} {2018})}\BibitemShut {NoStop}%
\bibitem [{\citenamefont {Sun}\ \emph {et~al.}(2017)\citenamefont {Sun},
  \citenamefont {Mao},\ and\ \citenamefont {Luo}}]{Luo2017B}%
  \BibitemOpen
  \bibfield  {author} {\bibinfo {author} {\bibfnamefont {Y.}~\bibnamefont
  {Sun}}, \bibinfo {author} {\bibfnamefont {Y.}~\bibnamefont {Mao}}, \ and\
  \bibinfo {author} {\bibfnamefont {S.}~\bibnamefont {Luo}},\ }\href
  {http://stacks.iop.org/0295-5075/118/i=6/a=60007} {\bibfield  {journal}
  {\bibinfo  {journal} {EPL}\ }\textbf {\bibinfo {volume}
  {118}},\ \bibinfo {pages} {60007} (\bibinfo {year} {2017})}\BibitemShut
  {NoStop}%
\bibitem [{\citenamefont {Spehner}\ and\ \citenamefont
  {Orszag}(2013)}]{spehner2013A}%
  \BibitemOpen
  \bibfield  {author} {\bibinfo {author} {\bibfnamefont {D.}~\bibnamefont
  {Spehner}}\ and\ \bibinfo {author} {\bibfnamefont {M.}~\bibnamefont
  {Orszag}},\ }\href {http://stacks.iop.org/1367-2630/15/i=10/a=103001}
  {\bibfield  {journal} {\bibinfo  {journal} {New J. Phys.}\ }\textbf
  {\bibinfo {volume} {15}},\ \bibinfo {pages} {103001} (\bibinfo {year}
  {2013})}\BibitemShut {NoStop}%
\bibitem [{\citenamefont {Xiong}\ and\ \citenamefont {Wu}(2018)}]{Xiong2018A}%
  \BibitemOpen
  \bibfield  {author} {\bibinfo {author} {\bibfnamefont {C.}~\bibnamefont
  {Xiong}}\ and\ \bibinfo {author} {\bibfnamefont {J.}~\bibnamefont {Wu}},\
  }\href {http://iopscience.iop.org/10.1088/1751-8121/aac979} {\bibfield
  {journal} {\bibinfo  {journal} {J. Phys. A: Math.
  Theor.}\ } \textbf
  {\bibinfo {volume} {51}},\ \bibinfo {pages}{414005} (\bibinfo {year} {2018})}\BibitemShut {NoStop}%
\bibitem [{\citenamefont {Bennett}\ \emph {et~al.}(1996)\citenamefont
  {Bennett}, \citenamefont {DiVincenzo}, \citenamefont {Smolin},\ and\
  \citenamefont {Wootters}}]{Bennett1996}%
  \BibitemOpen
  \bibfield  {author} {\bibinfo {author} {\bibfnamefont {C.~H.}\ \bibnamefont
  {Bennett}}, \bibinfo {author} {\bibfnamefont {D.~P.}\ \bibnamefont
  {DiVincenzo}}, \bibinfo {author} {\bibfnamefont {J.~A.}\ \bibnamefont
  {Smolin}}, \ and\ \bibinfo {author} {\bibfnamefont {W.~K.}\ \bibnamefont
  {Wootters}},\ }\href {\doibase 10.1103/PhysRevA.54.3824} {\bibfield
  {journal} {\bibinfo  {journal} {Phys. Rev. A}\ }\textbf {\bibinfo {volume}
  {54}},\ \bibinfo {pages} {3824} (\bibinfo {year} {1996})}\BibitemShut
  {NoStop}%
\bibitem [{\citenamefont {Horodecki}\ \emph {et~al.}(1998)\citenamefont
  {Horodecki}, \citenamefont {Horodecki},\ and\ \citenamefont
  {Horodecki}}]{Horodecki1998}%
  \BibitemOpen
  \bibfield  {author} {\bibinfo {author} {\bibfnamefont {M.}~\bibnamefont
  {Horodecki}}, \bibinfo {author} {\bibfnamefont {P.}~\bibnamefont
  {Horodecki}}, \ and\ \bibinfo {author} {\bibfnamefont {R.}~\bibnamefont
  {Horodecki}},\ }\href {\doibase 10.1103/PhysRevLett.80.5239} {\bibfield
  {journal} {\bibinfo  {journal} {Phys. Rev. Lett.}\ }\textbf {\bibinfo
  {volume} {80}},\ \bibinfo {pages} {5239} (\bibinfo {year}
  {1998})}\BibitemShut {NoStop}%
\bibitem [{\citenamefont {Rains}(1999{\natexlab{a}})}]{Rains1999A}%
  \BibitemOpen
  \bibfield  {author} {\bibinfo {author} {\bibfnamefont {E.~M.}\ \bibnamefont
  {Rains}},\ }\href {\doibase 10.1103/PhysRevA.60.173} {\bibfield  {journal}
  {\bibinfo  {journal} {Phys. Rev. A}\ }\textbf {\bibinfo {volume} {60}},\
  \bibinfo {pages} {173} (\bibinfo {year} {1999}{\natexlab{a}})}\BibitemShut
  {NoStop}%
\bibitem [{\citenamefont {Rains}(1999{\natexlab{b}})}]{Rains1999B}%
  \BibitemOpen
  \bibfield  {author} {\bibinfo {author} {\bibfnamefont {E.~M.}\ \bibnamefont
  {Rains}},\ }\href {\doibase 10.1103/PhysRevA.60.179} {\bibfield  {journal}
  {\bibinfo  {journal} {Phys. Rev. A}\ }\textbf {\bibinfo {volume} {60}},\
  \bibinfo {pages} {179} (\bibinfo {year} {1999}{\natexlab{b}})}\BibitemShut
  {NoStop}%
\bibitem [{\citenamefont {Luo}\ and\ \citenamefont {Zhang}(2004)}]{luo2004}%
  \BibitemOpen
  \bibfield  {author} {\bibinfo {author} {\bibfnamefont {S.}~\bibnamefont
  {Luo}}\ and\ \bibinfo {author} {\bibfnamefont {Q.}~\bibnamefont {Zhang}},\
  }\href {\doibase 10.1103/PhysRevA.69.032106} {\bibfield  {journal} {\bibinfo
  {journal} {Phys. Rev. A}\ }\textbf {\bibinfo {volume} {69}},\ \bibinfo
  {pages} {032106} (\bibinfo {year} {2004})}\BibitemShut {NoStop}%
\bibitem [{\citenamefont {BHATTACHARYYA}(1943)}]{Bhattacharyya1943}%
  \BibitemOpen
  \bibfield  {author} {\bibinfo {author} {\bibfnamefont {A.}~\bibnamefont
  {Bhattacharyya}},\ }\href {https://ci.nii.ac.jp/naid/10030997340/en/}
  {\bibfield  {journal} {\bibinfo  {journal} {Bull. Calcutta Math. Soc.}\
  }\textbf {\bibinfo {volume} {35}},\ \bibinfo {pages} {99} (\bibinfo {year}
  {1943})}\BibitemShut {NoStop}%
\bibitem [{\citenamefont {Vedral}\ and\ \citenamefont
  {Plenio}(1998)}]{vedral1998}%
  \BibitemOpen
  \bibfield  {author} {\bibinfo {author} {\bibfnamefont {V.}~\bibnamefont
  {Vedral}}\ and\ \bibinfo {author} {\bibfnamefont {M.~B.}\ \bibnamefont
  {Plenio}},\ }\href {\doibase 10.1103/PhysRevA.57.1619} {\bibfield  {journal}
  {\bibinfo  {journal} {Phys. Rev. A}\ }\textbf {\bibinfo {volume} {57}},\
  \bibinfo {pages} {1619} (\bibinfo {year} {1998})}\BibitemShut {NoStop}%
\bibitem{note2} Supposing a CPTP map $\Phi$ has Kraus operators $\{K_i\}$, we have
\begin{align*}
&\sum_i X(K_i\rho K_i^{\dagger},K_i\sigma K^{\dagger}_i)\nonumber\\
\overset{(ST)}=&\sum_i X(\mathrm{Tr}_E(I\otimes\ket{i}\bra{i}U\rho\otimes\ket{0}\bra{0}U^{\dagger}I\otimes\ket{i}\bra{i}),\\
&\mathrm{Tr}_E(I\otimes\ket{i}\bra{i}U\sigma\otimes\ket{0}\bra{0}U^{\dagger}I\otimes\ket{i}\bra{i}))\\
\overset{(P3)}\ge&\sum_i X(I\otimes\ket{i}\bra{i}U\rho\otimes\ket{0}\bra{0}U^{\dagger}I\otimes\ket{i}\bra{i},\\
&I\otimes\ket{i}\bra{i}U\sigma\otimes\ket{0}\bra{0}U^{\dagger}I\otimes\ket{i}\bra{i})\\
\overset{(P4)}=& X\big(\sum_i I\otimes\ket{i}\bra{i}U\rho\otimes\ket{0}\bra{0}U^{\dagger}I\otimes\ket{i}\bra{i},\\
&\sum_i I\otimes\ket{i}\bra{i}U\sigma\otimes\ket{0}\bra{0}U^{\dagger}I\otimes\ket{i}\bra{i} \big)\\
\overset{(P3)}\ge&X(U\rho\otimes\ket{0}\bra{0}U^{\dagger},U\sigma\otimes\ket{0}\bra{0}U^{\dagger})\\
=&X(\rho,\sigma).
\end{align*}
In the above, the first equality is due to the Stinespring theorem (ST) which states that there exists an extended Hilbert space $\mathcal{H}_E$, a pure state $\ket{0}\in\mathcal{H}_E$, a set of complete projectors $\{P_i\}$ and a global unitary $U$ such that $K_i\rho K_i^{\dagger}=\mathrm{Tr}_E(I\otimes P_iU\rho\otimes\ket{0}\bra{0}U^{\dagger}I\otimes P_i)$. The second equality follows from property (P4). On the other hand, the first and second inequalities are due to monotonicity (P3).
\bibitem [{\citenamefont {Wehrl}(1978)}]{Wehrl1978}%
  \BibitemOpen
  \bibfield  {author} {\bibinfo {author} {\bibfnamefont {A.}~\bibnamefont
  {Wehrl}},\ }\href {\doibase 10.1103/RevModPhys.50.221} {\bibfield  {journal}
  {\bibinfo  {journal} {Rev. Mod. Phys.}\ }\textbf {\bibinfo {volume} {50}},\
  \bibinfo {pages} {221} (\bibinfo {year} {1978})}\BibitemShut {NoStop}%
\bibitem [{\citenamefont {Vedral}(2002)}]{Vedral2002}%
  \BibitemOpen
  \bibfield  {author} {\bibinfo {author} {\bibfnamefont {V.}~\bibnamefont
  {Vedral}},\ }\href {\doibase 10.1103/RevModPhys.74.197} {\bibfield  {journal}
  {\bibinfo  {journal} {Rev. Mod. Phys.}\ }\textbf {\bibinfo {volume} {74}},\
  \bibinfo {pages} {197} (\bibinfo {year} {2002})}\BibitemShut {NoStop}%
\bibitem [{\citenamefont {Lindblad}(1974)}]{Lindblad1974}%
  \BibitemOpen
  \bibfield  {author} {\bibinfo {author} {\bibfnamefont {G.}~\bibnamefont
  {Lindblad}},\ }\href {\doibase 10.1007/BF01608390} {\bibfield  {journal}
  {\bibinfo  {journal} {Comm. Math. Phys.}\ }\textbf
  {\bibinfo {volume} {39}},\ \bibinfo {pages} {111} (\bibinfo {year}
  {1974})}\BibitemShut {NoStop}%
\bibitem [{\citenamefont {Henderson}\ and\ \citenamefont
  {Vedral}(2000)}]{Henderson2000}%
  \BibitemOpen
  \bibfield  {author} {\bibinfo {author} {\bibfnamefont {L.}~\bibnamefont
  {Henderson}}\ and\ \bibinfo {author} {\bibfnamefont {V.}~\bibnamefont
  {Vedral}},\ }\href {\doibase 10.1103/PhysRevLett.84.2263} {\bibfield
  {journal} {\bibinfo  {journal} {Phys. Rev. Lett.}\ }\textbf {\bibinfo
  {volume} {84}},\ \bibinfo {pages} {2263} (\bibinfo {year}
  {2000})}\BibitemShut {NoStop}%
\bibitem [{\citenamefont {Modi}\ \emph {et~al.}(2010)\citenamefont {Modi},
  \citenamefont {Paterek}, \citenamefont {Son}, \citenamefont {Vedral},\ and\
  \citenamefont {Williamson}}]{Modi2010}%
  \BibitemOpen
  \bibfield  {author} {\bibinfo {author} {\bibfnamefont {K.}~\bibnamefont
  {Modi}}, \bibinfo {author} {\bibfnamefont {T.}~\bibnamefont {Paterek}},
  \bibinfo {author} {\bibfnamefont {W.}~\bibnamefont {Son}}, \bibinfo {author}
  {\bibfnamefont {V.}~\bibnamefont {Vedral}}, \ and\ \bibinfo {author}
  {\bibfnamefont {M.}~\bibnamefont {Williamson}},\ }\href {\doibase
  10.1103/PhysRevLett.104.080501} {\bibfield  {journal} {\bibinfo  {journal}
  {Phys. Rev. Lett.}\ }\textbf {\bibinfo {volume} {104}},\ \bibinfo {pages}
  {080501} (\bibinfo {year} {2010})}\BibitemShut {NoStop}%
\bibitem [{\citenamefont {Yao}\ \emph {et~al.}(2015)\citenamefont {Yao},
  \citenamefont {Xiao}, \citenamefont {Ge},\ and\ \citenamefont
  {Sun}}]{Yao2015}%
  \BibitemOpen
  \bibfield  {author} {\bibinfo {author} {\bibfnamefont {Y.}~\bibnamefont
  {Yao}}, \bibinfo {author} {\bibfnamefont {X.}~\bibnamefont {Xiao}}, \bibinfo
  {author} {\bibfnamefont {L.}~\bibnamefont {Ge}}, \ and\ \bibinfo {author}
  {\bibfnamefont {C.~P.}\ \bibnamefont {Sun}},\ }\href {\doibase
  10.1103/PhysRevA.92.022112} {\bibfield  {journal} {\bibinfo  {journal} {Phys.
  Rev. A}\ }\textbf {\bibinfo {volume} {92}},\ \bibinfo {pages} {022112}
  (\bibinfo {year} {2015})}\BibitemShut {NoStop}%
\bibitem [{\citenamefont {Streltsov}\ \emph {et~al.}(2010)\citenamefont
  {Streltsov}, \citenamefont {Kampermann},\ and\ \citenamefont
  {Bru{\ss}}}]{streltsov2010}%
  \BibitemOpen
  \bibfield  {author} {\bibinfo {author} {\bibfnamefont {A.}~\bibnamefont
  {Streltsov}}, \bibinfo {author} {\bibfnamefont {H.}~\bibnamefont
  {Kampermann}}, \ and\ \bibinfo {author} {\bibfnamefont {D.}~\bibnamefont
  {Bru{\ss}}},\ }\href {http://stacks.iop.org/1367-2630/12/i=12/a=123004}
  {\bibfield  {journal} {\bibinfo  {journal} {New J. Phys.}\ }\textbf
  {\bibinfo {volume} {12}},\ \bibinfo {pages} {123004} (\bibinfo {year}
  {2010})}\BibitemShut {NoStop}%
 %
\bibitem{note3} We have, using Eq. (\ref{eq1}),
\begin{align}
R_F(\rho)&:=\min_{\sigma\in\mathcal{FS}}d_F(\rho,\sigma)=1-\max_{\sigma\in\mathcal{FS}}F^2(\rho,\sigma) \nonumber \\
&\overset{Eq. (\ref{eq1})}=1-\max_{\rho=\sum_ip_i\rho_i}\sum_ip_i\max_{\sigma_i\in\mathcal{FS}}F^2(\rho_i,\sigma_i) \nonumber \\
&=\min_{\rho=\sum_ip_i\rho_i} \sum_ip_i \left(1-\max_{\sigma_i\in\mathcal{FS}}F^2(\rho_i,\sigma_i) \right) \nonumber \\
&=\min_{\rho=\sum_ip_i\rho_i}\sum_i p_i R_F(\rho_i).\nonumber
\end{align}
 %
%
\bibitem [{\citenamefont {Yu}(2017)}]{yucs2017}%
  \BibitemOpen
  \bibfield  {author} {\bibinfo {author} {\bibfnamefont {C. S.}\ \bibnamefont
  {Yu}},\ }\href {\doibase 10.1103/PhysRevA.95.042337} {\bibfield  {journal}
  {\bibinfo  {journal} {Phys. Rev. A}\ }\textbf {\bibinfo {volume} {95}},\
  \bibinfo {pages} {042337} (\bibinfo {year} {2017})}\BibitemShut {NoStop}%
\bibitem{Streltsov2015B} A. Streltsov, U. Singh, H. S. Dhar, M. N. Bera, and G. Adesso,
\href {https://link.aps.org/doi/10.1103/PhysRevLett.115.020403} {Phys. Rev. Lett. {\bf 115}, 020403 (2015)}.
\bibitem{Xiong2018B} C. Xiong, A. Kumar, J. Wu,
\href {https://link.aps.org/doi/10.1103/PhysRevA.98.032324} {Phys. Rev. A {\bf 98}, 032324 (2018)}.
\bibitem [{\citenamefont {Helstrom}(1976)}]{Helstrom1976}%
  \BibitemOpen
  \bibfield  {author} {\bibinfo {author} {\bibfnamefont {C. W.}\ \bibnamefont
  {Helstrom}},\ }\href {http://cds.cern.ch/record/110988} {\emph {\bibinfo
  {title} {Quantum Detection and Estimation Theory}}}\ (\bibinfo  {publisher}
  {New York, NY : Academic Press},\ \bibinfo {year} {1976})\BibitemShut
  {NoStop}%
\bibitem [{\citenamefont {Bergou}\ \emph {et~al.}(2004)\citenamefont {Bergou},
  \citenamefont {Herzog},\ and\ \citenamefont {Hillery}}]{Bergou2004}%
  \BibitemOpen
  \bibfield  {author} {\bibinfo {author} {\bibfnamefont {J.~A.}\ \bibnamefont
  {Bergou}}, \bibinfo {author} {\bibfnamefont {U.}~\bibnamefont {Herzog}}, \
  and\ \bibinfo {author} {\bibfnamefont {M.}~\bibnamefont {Hillery}},\
  }\enquote {\bibinfo {title} {11 discrimination of quantum states},}\ in\
  \href {\doibase 10.1007/978-3-540-44481-7_11} {\emph {\bibinfo {booktitle}
  {Quantum State Estimation}}},\ \bibinfo {editor} {edited by\ \bibinfo
  {editor} {\bibfnamefont {M.}~\bibnamefont {Paris}}\ and\ \bibinfo {editor}
  {\bibfnamefont {J.}~\bibnamefont {{\v{R}}eh{\'a}{\v{c}}ek}}}\ (\bibinfo
  {publisher} {Springer Berlin Heidelberg},\ \bibinfo {address} {Berlin,
  Heidelberg},\ \bibinfo {year} {2004})\ pp.\ \bibinfo {pages}
  {417--465}\BibitemShut {NoStop}%
\bibitem [{\citenamefont {Y.C.~Eldar}(2004)}]{Eldar2004}%
  \BibitemOpen
  \bibfield  {author} {\bibinfo {author} {\bibfnamefont {G.~V.}\ \bibnamefont
  {Y.C.~Eldar}, \bibfnamefont {A.~Megretski}},\ }\href {\doibase
  10.1109/TIT.2004.828070} {\bibfield  {journal} {\bibinfo  {journal} {IEEE
  Transactions on Information Theory}\ }\textbf {\bibinfo {volume} {50}},\
  \bibinfo {pages} {1198 } (\bibinfo {year} {2004})}\BibitemShut {NoStop}%
\bibitem [{\citenamefont {Chou}\ and\ \citenamefont {Hsu}(2003)}]{Chou2003}%
  \BibitemOpen
  \bibfield  {author} {\bibinfo {author} {\bibfnamefont {C.-L.}\ \bibnamefont
  {Chou}}\ and\ \bibinfo {author} {\bibfnamefont {L.~Y.}\ \bibnamefont {Hsu}},\
  }\href {\doibase 10.1103/PhysRevA.68.042305} {\bibfield  {journal} {\bibinfo
  {journal} {Phys. Rev. A}\ }\textbf {\bibinfo {volume} {68}},\ \bibinfo
  {pages} {042305} (\bibinfo {year} {2003})}\BibitemShut {NoStop}%
\bibitem [{\citenamefont {BELAVKIN}(1975)}]{Belavkin1975a}%
  \BibitemOpen
  \bibfield  {author} {\bibinfo {author} {\bibfnamefont {V.~P.}\ \bibnamefont
  {BELAVKIN}},\ }\href {https://ci.nii.ac.jp/naid/20001256994/en/} {\bibfield
  {journal} {\bibinfo  {journal} {Radiotekhnika i Elektronika}\ }\textbf
  {\bibinfo {volume} {20}},\ \bibinfo {pages} {1177} (\bibinfo {year}
  {1975})}\BibitemShut {NoStop}%
\bibitem [{\citenamefont {Belavkin}(1975)}]{Belavkin1975}%
  \BibitemOpen
  \bibfield  {author} {\bibinfo {author} {\bibfnamefont {V. P.}\ \bibnamefont
  {Belavkin}},\ }\href {\doibase 10.1080/17442507508833114} {\bibfield
  {journal} {\bibinfo  {journal} {Stochastics}\ }\textbf {\bibinfo {volume}
  {1}},\ \bibinfo {pages} {315} (\bibinfo {year} {1975})},\ \Eprint
  {http://arxiv.org/abs/https://doi.org/10.1080/17442507508833114}
  {https://doi.org/10.1080/17442507508833114} \BibitemShut {NoStop}%
\bibitem [{\citenamefont {Holevo}(1978)}]{Holevo1978}%
  \BibitemOpen
  \bibfield  {author} {\bibinfo {author} {\bibfnamefont {A.}~\bibnamefont
  {Holevo}},\ }\bibfield  {booktitle} {\emph {\bibinfo {booktitle} {Teoriya
  Veroyatnostej i Ee Primeneniya}},\ }\href@noop {} {\ \textbf {\bibinfo
  {volume} {23}} (\bibinfo {year} {1978})}\BibitemShut {NoStop}%
\bibitem [{\citenamefont {Hausladen}\ and\ \citenamefont
  {Wootters}(1994)}]{Hausladen1994}%
  \BibitemOpen
  \bibfield  {author} {\bibinfo {author} {\bibfnamefont {P.}~\bibnamefont
  {Hausladen}}\ and\ \bibinfo {author} {\bibfnamefont {W.~K.}\ \bibnamefont
  {Wootters}},\ }\href {\doibase 10.1080/09500349414552221} {\bibfield
  {journal} {\bibinfo  {journal} {J. Mod. Opt.}\ }\textbf {\bibinfo
  {volume} {41}},\ \bibinfo {pages} {2385} (\bibinfo {year} {1994})},\ \Eprint
  {http://arxiv.org/abs/https://doi.org/10.1080/09500349414552221}
  {https://doi.org/10.1080/09500349414552221} \BibitemShut {NoStop}%
\bibitem [{\citenamefont {Hausladen}\ \emph {et~al.}(1996)\citenamefont
  {Hausladen}, \citenamefont {Jozsa}, \citenamefont {Schumacher}, \citenamefont
  {Westmoreland},\ and\ \citenamefont {Wootters}}]{Hausladen1996}%
  \BibitemOpen
  \bibfield  {author} {\bibinfo {author} {\bibfnamefont {P.}~\bibnamefont
  {Hausladen}}, \bibinfo {author} {\bibfnamefont {R.}~\bibnamefont {Jozsa}},
  \bibinfo {author} {\bibfnamefont {B.}~\bibnamefont {Schumacher}}, \bibinfo
  {author} {\bibfnamefont {M.}~\bibnamefont {Westmoreland}}, \ and\ \bibinfo
  {author} {\bibfnamefont {W.~K.}\ \bibnamefont {Wootters}},\ }\href {\doibase
  10.1103/PhysRevA.54.1869} {\bibfield  {journal} {\bibinfo  {journal} {Phys.
  Rev. A}\ }\textbf {\bibinfo {volume} {54}},\ \bibinfo {pages} {1869}
  (\bibinfo {year} {1996})}\BibitemShut {NoStop}%
\bibitem [{\citenamefont {Peres}\ and\ \citenamefont
  {Wootters}(1991)}]{peres1991}%
  \BibitemOpen
  \bibfield  {author} {\bibinfo {author} {\bibfnamefont {A.}~\bibnamefont
  {Peres}}\ and\ \bibinfo {author} {\bibfnamefont {W.~K.}\ \bibnamefont
  {Wootters}},\ }\href {\doibase 10.1103/PhysRevLett.66.1119} {\bibfield
  {journal} {\bibinfo  {journal} {Phys. Rev. Lett.}\ }\textbf {\bibinfo
  {volume} {66}},\ \bibinfo {pages} {1119} (\bibinfo {year}
  {1991})}\BibitemShut {NoStop}%
\bibitem [{\citenamefont {Eldar}\ and\ \citenamefont
  {Forney}(2001)}]{Eldar2001}%
  \BibitemOpen
  \bibfield  {author} {\bibinfo {author} {\bibfnamefont {Y.~C.}\ \bibnamefont
  {Eldar}}\ and\ \bibinfo {author} {\bibfnamefont {G.~D.}\ \bibnamefont
  {Forney}},\ }\href {\doibase 10.1109/18.915636} {\bibfield  {journal}
  {\bibinfo  {journal} {IEEE Transactions on Information Theory}\ }\textbf
  {\bibinfo {volume} {47}},\ \bibinfo {pages} {858} (\bibinfo {year}
  {2001})}\BibitemShut {NoStop}%
\bibitem [{\citenamefont {Spehner}(2014)}]{Spehner2014}%
  \BibitemOpen
  \bibfield  {author} {\bibinfo {author} {\bibfnamefont {D.}~\bibnamefont
  {Spehner}},\ }\href {\doibase 10.1063/1.4885832} {\bibfield  {journal}
  {\bibinfo  {journal} {J. Math. Phys.}\ }\textbf {\bibinfo
  {volume} {55}},\ \bibinfo {pages} {075211} (\bibinfo {year} {2014})} \BibitemShut {NoStop}%
\bibitem [{\citenamefont {Peres}(1990)}]{Peres1990}%
  \BibitemOpen
  \bibfield  {author} {\bibinfo {author} {\bibfnamefont {A.}~\bibnamefont
  {Peres}},\ }\href {\doibase 10.1007/BF01883517} {\bibfield  {journal}
  {\bibinfo  {journal} {Found. Phys.}\ }\textbf {\bibinfo {volume}
  {20}},\ \bibinfo {pages} {1441} (\bibinfo {year} {1990})}\BibitemShut
  {NoStop}%
\bibitem [{\citenamefont {S.Kennedy}(1973)}]{kennedy1973}%
  \BibitemOpen
  \bibfield  {author} {\bibinfo {author} {\bibfnamefont {R.}~\bibnamefont
  {S.Kennedy}},\ }\href@noop {} {\bibfield  {journal} {\bibinfo  {journal} {MIT
  Research Laboratory Electronic Quarterly Progress Report,Technical Report}\
  }\textbf {\bibinfo {volume} {110}},\ \bibinfo {pages} {142} (\bibinfo {year}
  {1973})}\BibitemShut {NoStop}%
\bibitem [{\citenamefont {Eldar}(2003)}]{Eldar2003}%
  \BibitemOpen
  \bibfield  {author} {\bibinfo {author} {\bibfnamefont {Y.~C.}\ \bibnamefont
  {Eldar}},\ }\href {\doibase 10.1103/PhysRevA.68.052303} {\bibfield  {journal}
  {\bibinfo  {journal} {Phys. Rev. A}\ }\textbf {\bibinfo {volume} {68}},\
  \bibinfo {pages} {052303} (\bibinfo {year} {2003})}\BibitemShut {NoStop}%
\bibitem [{\citenamefont {Vedral}\ \emph {et~al.}(1997)\citenamefont {Vedral},
  \citenamefont {Plenio}, \citenamefont {Rippin},\ and\ \citenamefont
  {Knight}}]{Vedral1997}%
  \BibitemOpen
  \bibfield  {author} {\bibinfo {author} {\bibfnamefont {V.}~\bibnamefont
  {Vedral}}, \bibinfo {author} {\bibfnamefont {M.~B.}\ \bibnamefont {Plenio}},
  \bibinfo {author} {\bibfnamefont {M.~A.}\ \bibnamefont {Rippin}}, \ and\
  \bibinfo {author} {\bibfnamefont {P.~L.}\ \bibnamefont {Knight}},\ }\href
  {\doibase 10.1103/PhysRevLett.78.2275} {\bibfield  {journal} {\bibinfo
  {journal} {Phys. Rev. Lett.}\ }\textbf {\bibinfo {volume} {78}},\ \bibinfo
  {pages} {2275} (\bibinfo {year} {1997})}\BibitemShut {NoStop}%
\bibitem [{\citenamefont {Luo}(2008)}]{luo2008A}%
  \BibitemOpen
  \bibfield  {author} {\bibinfo {author} {\bibfnamefont {S.}~\bibnamefont
  {Luo}},\ }\href {\doibase 10.1103/PhysRevA.77.042303} {\bibfield  {journal}
  {\bibinfo  {journal} {Phys. Rev. A}\ }\textbf {\bibinfo {volume} {77}},\
  \bibinfo {pages} {042303} (\bibinfo {year} {2008})}\BibitemShut {NoStop}%
\bibitem [{\citenamefont {Daki\ifmmode~\acute{c}\else \'{c}\fi{}}\ \emph
  {et~al.}(2010)\citenamefont {Daki\ifmmode~\acute{c}\else \'{c}\fi{}},
  \citenamefont {Vedral},\ and\ \citenamefont {Brukner}}]{Dakic2010}%
  \BibitemOpen
  \bibfield  {author} {\bibinfo {author} {\bibfnamefont {B.}~\bibnamefont
  {Daki\ifmmode~\acute{c}\else \'{c}\fi{}}}, \bibinfo {author} {\bibfnamefont
  {V.}~\bibnamefont {Vedral}}, \ and\ \bibinfo {author} {\bibfnamefont
  {C.}~\bibnamefont {Brukner}},\ }\href {\doibase
  10.1103/PhysRevLett.105.190502} {\bibfield  {journal} {\bibinfo  {journal}
  {Phys. Rev. Lett.}\ }\textbf {\bibinfo {volume} {105}},\ \bibinfo {pages}
  {190502} (\bibinfo {year} {2010})}\BibitemShut {NoStop}%
\bibitem [{\citenamefont {Ali}\ \emph {et~al.}(2010)\citenamefont {Ali},
  \citenamefont {Rau},\ and\ \citenamefont {Alber}}]{M.Ali2010A}%
  \BibitemOpen
  \bibfield  {author} {\bibinfo {author} {\bibfnamefont {M.}~\bibnamefont
  {Ali}}, \bibinfo {author} {\bibfnamefont {A.~R.~P.}\ \bibnamefont {Rau}}, \
  and\ \bibinfo {author} {\bibfnamefont {G.}~\bibnamefont {Alber}},\ }\href
  {\doibase 10.1103/PhysRevA.81.042105} {\bibfield  {journal} {\bibinfo
  {journal} {Phys. Rev. A}\ }\textbf {\bibinfo {volume} {81}},\ \bibinfo
  {pages} {042105} (\bibinfo {year} {2010})}\BibitemShut {NoStop}%
\bibitem [{\citenamefont {Rau}(2009)}]{Rau2009}%
  \BibitemOpen
  \bibfield  {author} {\bibinfo {author} {\bibfnamefont {A.~R.~P.}\
  \bibnamefont {Rau}},\ }\href
  {http://stacks.iop.org/1751-8121/42/i=41/a=412002} {\bibfield  {journal}
  {\bibinfo  {journal} {J.Phys. A: Math. Theor.}\
  }\textbf {\bibinfo {volume} {42}},\ \bibinfo {pages} {412002} (\bibinfo
  {year} {2009})}\BibitemShut {NoStop}%
\bibitem [{\citenamefont {Bromley}\ \emph {et~al.}(2015)\citenamefont
  {Bromley}, \citenamefont {Cianciaruso},\ and\ \citenamefont
  {Adesso}}]{Bromley2015}%
  \BibitemOpen
  \bibfield  {author} {\bibinfo {author} {\bibfnamefont {T.~R.}\ \bibnamefont
  {Bromley}}, \bibinfo {author} {\bibfnamefont {M.}~\bibnamefont
  {Cianciaruso}}, \ and\ \bibinfo {author} {\bibfnamefont {G.}~\bibnamefont
  {Adesso}},\ }\href {\doibase 10.1103/PhysRevLett.114.210401} {\bibfield
  {journal} {\bibinfo  {journal} {Phys. Rev. Lett.}\ }\textbf {\bibinfo
  {volume} {114}},\ \bibinfo {pages} {210401} (\bibinfo {year}
  {2015})}\BibitemShut {NoStop}%
\bibitem [{\citenamefont {Wigner}(1963)}]{wigner1963}%
  \BibitemOpen
  \bibfield  {author} {\bibinfo {author} {\bibfnamefont {E.~P.}\ \bibnamefont
  {Wigner}},\ }\href {https://ci.nii.ac.jp/naid/20001314484/en/} {\bibfield
  {journal} {\bibinfo  {journal} {Proc. Nat. Acad. Sci. U.S.A.}\ }\textbf
  {\bibinfo {volume} {49}},\ \bibinfo {pages} {910} (\bibinfo {year}
  {1963})}\BibitemShut {NoStop}%
\bibitem [{\citenamefont {Luo}(2003)}]{Luo2003}%
  \BibitemOpen
  \bibfield  {author} {\bibinfo {author} {\bibfnamefont {S.}~\bibnamefont
  {Luo}},\ }\href {\doibase 10.1103/PhysRevLett.91.180403} {\bibfield
  {journal} {\bibinfo  {journal} {Phys. Rev. Lett.}\ }\textbf {\bibinfo
  {volume} {91}},\ \bibinfo {pages} {180403} (\bibinfo {year}
  {2003})}\BibitemShut {NoStop}%
\bibitem [{\citenamefont {Tan}\ \emph {et~al.}(2016)\citenamefont {Tan},
  \citenamefont {Kwon}, \citenamefont {Park},\ and\ \citenamefont
  {Jeong}}]{Tan2016}%
  \BibitemOpen
  \bibfield  {author} {\bibinfo {author} {\bibfnamefont {K.~C.}\ \bibnamefont
  {Tan}}, \bibinfo {author} {\bibfnamefont {H.}~\bibnamefont {Kwon}}, \bibinfo
  {author} {\bibfnamefont {C.-Y.}\ \bibnamefont {Park}}, \ and\ \bibinfo
  {author} {\bibfnamefont {H.}~\bibnamefont {Jeong}},\ }\href {\doibase
  10.1103/PhysRevA.94.022329} {\bibfield  {journal} {\bibinfo  {journal} {Phys.
  Rev. A}\ }\textbf {\bibinfo {volume} {94}},\ \bibinfo {pages} {022329}
  (\bibinfo {year} {2016})}\BibitemShut {NoStop}%
\bibitem [{\citenamefont {Tan}\ and\ \citenamefont {Jeong}(2018)}]{Tan2018}%
  \BibitemOpen
  \bibfield  {author} {\bibinfo {author} {\bibfnamefont {K.~C.}\ \bibnamefont
  {Tan}}\ and\ \bibinfo {author} {\bibfnamefont {H.}~\bibnamefont {Jeong}},\
  }\href {\doibase 10.1103/PhysRevLett.121.220401} {\bibfield  {journal}
  {\bibinfo  {journal} {Phys. Rev. Lett.}\ }\textbf {\bibinfo {volume} {121}},\
  \bibinfo {pages} {220401} (\bibinfo {year} {2018})}\BibitemShut {NoStop}%
\bibitem [{\citenamefont {Nielsen}(1999)}]{Nielsen1999}%
  \BibitemOpen
  \bibfield  {author} {\bibinfo {author} {\bibfnamefont {M.~A.}\ \bibnamefont
  {Nielsen}},\ }\href {\doibase 10.1103/PhysRevLett.83.436} {\bibfield
  {journal} {\bibinfo  {journal} {Phys. Rev. Lett.}\ }\textbf {\bibinfo
  {volume} {83}},\ \bibinfo {pages} {436} (\bibinfo {year} {1999})}\BibitemShut
  {NoStop}%
\end{thebibliography}

%

\end{document}